\documentclass[a4,12pt]{article}
\usepackage{amsmath,amsthm}
\usepackage[english]{babel}
\usepackage{amssymb,amstext,amsfonts}
\usepackage{dsfont}
\newtheorem{proposition}{Proposition}[section]
\newtheorem{theo}{Theorem}[section]

\theoremstyle{definition}
\newtheorem{rem}[theo]{Remark}
\newtheorem{defi}[theo]{Definition}

\makeatletter
\@addtoreset{equation}{section}
\makeatother
\usepackage{color}

 \def\reel{ {I\!\!  R} }
\def\relatif{ {Z\hskip -5,1 pt Z\hskip 0,2 pt} }
\def\O{\Omega_q}

\def\d{\delta}
\def\D{\Delta}
\def\Dc{\Delta(c_{\Lambda,h}(T_c(\lambda),\lambda))}
\def\C{c_{\Lambda,h}(T,\lambda)}
\def\2{I$\!$I} \def\L{\Lambda}
\def\l{\langle}
\def\r{\rangle}
\def\H{H_\Lambda (c_{\Lambda,h}(T,\lambda))}
\def\I{ {1\hskip -4,5 pt 1\hskip 0,2 pt} }

\def\s{\sigma}
\def\g{\gamma}
\def\h{$\cal A$}
\def\a{\alpha}
\begin{document}

\title{ Bogoliubov quasi-averages: spontaneous symmetry breaking and algebra of fluctuations}

\author{Walter F. Wreszinski\\
        Instituto de Fisica USP\\
        Rua do Mat\~{a}o, s.n., Travessa R 187\\
        05508-090 S\~{a}o Paulo, Brazil\\
        \texttt{wreszins@gmail.com}\\
        and\\
        Valentin A. Zagrebnov\\
        D\'{e}partement de Math\'{e}matiques - AMU\\
        Institut de Math\'{e}matiques de Marseille (UMR 7373)\\
        CMI - Technop\^{o}le Ch\^{a}teau-Gombert\\
        39, Rue F. Joliot-Curie,
        13453 Marseille Cedex 13, France\\
        \texttt{valentin.zagrebnov@univ-amu.fr}\\
        }

\maketitle

\textit{In memory of Dmitry Nikolaevich Zubarev on the 100$^{\rm th}$ anniversary of his birthday}

\begin{abstract}
\noindent
The paper advocates the Bogoliubov method of quasi-averages for quantum systems.
First, we elucidate its applications to study the phase transitions with Spontaneous
Symmetry Breaking (SSB). To this aim we consider example of Bose-Einstein condensation (BEC) in continuous
systems. Our analysis of different type of generalised condensations demonstrates that the only
physically reliable quantities are those that defined by Bogoliubov quasi-averages.
In this connection we also give a solution of the problem posed by Lieb, Seiringer and Yngvason in
\cite{LSYng}.
Second, using the \textit{scaled} Bogoliubov method of quasi-averages and taking the structural quantum phase
transition as a basic example, we scrutinise a relation between SSB and the critical quantum fluctuations.
Our analysis shows that again the quasi-averages give an adequate tool for description of the algebra of
critical quantum fluctuation operators in the both commutative and noncommutative cases.
\end{abstract}

\newpage
\tableofcontents
\section{Introduction and summary}\label{sec:Intr-Summ}
The concept of Spontaneous Symmetry Breaking (SSB) is a central one in quantum physics, both in statistical
mechanics  and quantum field theory and particle physics.

The definition of SSB is well-known since the middle sixties
\cite{Ru}, Ch.6.5.2., as well as \cite{BR87}, Ch.4.3.4.
Recall that one starts from a state (ground or thermal), assumed to be invariant under a symmetry group
$G$, but which has a nontrivial decomposition into extremal states, which may be physically interpreted as
pure thermodynamic phases (states). The latter, however, do not exhibit invariance under $G$, but only under
a proper subgroup $H$ of $G$.

There are basically two ways of constructing extremal states: \\
(1) by a choice of boundary conditions (b.c.) for  Hamiltonians $H_{\Lambda}$ in finite regions
$\Lambda \subset \mathbb{R}^d$ and then take the thermodynamic limit ($\Lambda \uparrow  \mathbb{Z}^{d}$
or $\Lambda \uparrow  \mathbb{R}^{d}$) of expectations over the corresponding local states; \\
(2) by replacing: $H_{\Lambda} \rightarrow H_{\Lambda} + h \, B_{\Lambda}$, where $B_{\Lambda}$ is a
suitable extensive operator and $h$ a real parameter, then by taking first $\Lambda \uparrow \mathbb{Z}^{d}$
or $\Lambda \uparrow \mathbb{R}^{d}$, and the second, $h \to +0$ (or $h \to -0$). Here one assumes that the
states considered are locally normal or locally finite, see e.g. \cite{Sewell1} and references there.
The method (2) is known as Bogoliubov's \textit{quasi-averages} (q-a) method \cite{Bog07}-\cite{Bog70}.

We comment that although quite transparent, e.g. for classical lattice systems, the method of boundary
conditions is unsatisfactory for continuous and even worse for the quantum systems.
In this paper we advocate the Bogoliubov method of quasi-averages for quantum systems.

First, we elucidate its applications to study the phase transitions with SSB, see Section \ref{sec:Bosons}.
To this aim we consider first the quantum phase transition, which is the conventional one-mode
Bose-Einstein condensation (BEC) of the \textit{perfect} Bose-gas. In this simplest case the condensation occurs
in the \textit{single} zero mode, which implies a spontaneous breaking of $G$: the gauge group of
transformations (GSB). After that, we consider the case when the condensation is dispersed over
\textit{infinitely} many modes. Our analysis of different \textit{types} of
this \textit{generalised} condensation (gBEC) demonstrates that the only physically \textit{reliable} quantities,
are those that defined by the Bogoliubov method of q-a, see Remark \ref{rem:3.3} and Theorem \ref{theo:3.4}.

We extend this analysis to \textit{imperfect} Bose-gas. As a consequence of our results, a general question
posed by Lieb, Seiringer and Yngvason \cite{LSYng} concerning the equivalence between Bose-Einstein condensation
$\rm{(BEC)}_{q-a}$ and  Gauge Symmetry Breaking $\rm{(GSB)}_{q-a}$, both defined via the one-mode Bogoliubov
quasi-average, is elucidated for any {type} of {generalised} BEC (gBEC) \textit{\`{a} la} van
den Berg-Lewis-Pul\`{e} \cite{vdBLP} and \cite{BZ}, see Remark \ref{rem:3.5},
where it is also pointed out that the fact that quasi-averages lead to ergodic states clarifies an
important conceptual aspect of the quasi-average \textit{trick}.
Second, using the Bogoluibov method of q-a and taking the \textit{structural} quantum phase transition
as a basic example, we scrutinise a relation between SSB and the critical \textit{quantum fluctuations}, see
Section \ref{BQ-A-QFl}. Our analysis in Section \ref{Q-A-Cr-Q-Fl} shows that again the Bogoliubov quasi-averages
give an adequate tool for description of the algebra of \textit{fluctuation operators} on the critical line of
transitions. There we study the both commutative and noncommutative cases of this algebra, see
Theorem \ref{thm:5.1} and Theorem \ref{thm:5.2}.

We note here that it was Dmitry Nikolaevich Zubarev \cite{Zu70}, who for the first time indicated a
relevance of the Bogoliubov quasi-averages in the theory of non-equilibrium processes. In this case
the infinitesimal external sources serve to break the time-invariance of the Liouville equation for the
statistical operator. Although well-known in the mathematical physics as the limit-absorption principle this
approach was developed in \cite{Zu70} to many-body problems. This elegant extension is now called the
Zubarev method of a Non-equilibrium Statistical Operator \cite{Zu71}, \cite{ZMP}.

This interesting aspect of the Bogoliubov quasi-average method is out of the scope of the present paper.

\section{{Continuous boson systems}}\label{sec:Bosons}
\subsection{Conventional or generalised condensations and ODLRO} \label{sec:gBEC}
We note that existence of \textit{generalised} Bose-condensations (gBEC) makes the boson
systems more relevant for demonstration of efficiency of the Bogoliubov quasi-averages than, e.g., spin lattice
systems. This becomes clear even on the level of the Perfect Bose-gas (PBG).

To this aim we consider first the Bose-condensation of PBG in a three-dimensional anisotropic parallelepiped
$\Lambda:= V^{\alpha_1}\times V^{\alpha_2}\times V^{\alpha_3}$, with \textit{periodic} boundary
condition (p.b.c.) and $\alpha_1 \geq \alpha_2 \geq \alpha_3$, $\alpha_1 + \alpha_2 + \alpha_3 = 1$, i.e.
the volume $|\Lambda| = V$.  In the boson Fock space $\mathcal{F}_{\Lambda}:=
\mathcal{F}_{boson}(\mathcal{L}^2 (\Lambda))$ the Hamiltonian of this system for the grand-canonical ensemble
with chemical potential $\mu < 0$ is defined by :
\begin{eqnarray}\label{G-C-PBG}
H_{0,\Lambda,\mu} = T_{\Lambda} - \mu \, N_{\Lambda} =
\sum_{k \in \Lambda^{*}} (\varepsilon_{k} - \mu)\, b^{*}_{k} b_{k} \ , \ \ \ \
 {\rm{dom}}(H_{0,\Lambda,\mu})= {\rm{dom}}(T_{\Lambda}) \ .
\end{eqnarray}
Here one-particle kinetic-energy operator spectrum $\{\varepsilon_{k} = k^2\}_{k \in \Lambda^{*}}$, where
the set $\Lambda^{*}$ is dual to $\Lambda$:
\begin{equation}\label{dual-Lambda}
\Lambda ^{\ast }:= \{k_{j}= \frac{2\pi }{V^{{\alpha_{j}}}}n_{j} : n_{j}
\in \mathbb{Z} \}_{j=1}^{d=3}
\ \ \ {\rm{then}} \ \ \ \varepsilon _{k}= \sum_{j=1}^{d} {k_{j}^2} \ .
\end{equation}
We denote by $b_{k}:=b(\varphi_{k}^{\Lambda})$ and $b^{*}_{k}= (b(\varphi_{k}^{\Lambda}))^{*}$ the $k$-mode
boson annihilation and creation operators in the Fock space $\mathcal{F}_{\Lambda}$. They are indexed by
the ortho-normal basis $\{\varphi_{k}^{\Lambda}(x) = e^{i k x}/\sqrt{V}\}_{k \in \Lambda^{*}}$ in
$\mathcal{L}^2 (\Lambda)$ generated by
the eigenfunctions of the self-adjoint one-particle kinetic-energy operator $(- \Delta)_{p.b.c.}$ in
$\mathcal{L}^2 (\Lambda)$. Formally these operators satisfy the Canonical Commutation Relations (CCR):
$[b_{k},b^{*}_{k'}]=\delta_{k,k'}$, for $k, k' \in \Lambda^{*}$. Then $N_k =  b^{*}_{k} b_{k}$ is
occupation-number operator of the one-particle state $\varphi_{k}^{\Lambda}$ and $N_{\Lambda} =
\sum_{k \in \Lambda^{*}} N_k$ denotes the total-number operator in $\Lambda$.

For temperature $\beta^{-1} := k_B \, T$ and chemical potential $\mu$  we denote by
$\omega_{\beta,\mu,\Lambda}^{0}(\cdot)$ the grand-canonical Gibbs state of the PBG generated by (\ref{G-C-PBG}):
\begin{equation}\label{0-state}
\omega_{\beta,\mu,\Lambda}^{0}(\cdot) = \frac{{\rm{Tr}}_{{\mathcal{F}}_{\Lambda}}
(\exp(-\beta H_{0,\Lambda,\mu})\ \cdot \ )}
{{\rm{Tr}}_{{\mathcal{F}}_{\Lambda}} \exp(-\beta H_{0,\Lambda,\mu})} \ \ .
\end{equation}
Then the problem of existence of a Bose-condensation is related to solution of the equation
\begin{equation}\label{BEC-eq}
\rho = \frac{1}{V} \sum_{k \in \Lambda^{*}} \omega_{\beta,\mu,\Lambda}^{0}(N_k) =
\frac{1}{V} \sum_{k\in \Lambda ^{\ast}}\frac{1}{e^{\beta \left(\varepsilon_{k}-\mu \right)}-1} \ ,
\end{equation}
for a given total particle density  $\rho$ in $\Lambda$. Note that by (\ref{dual-Lambda}) the thermodynamic limit
$\Lambda \uparrow \mathbb{R}^3$ in the right-hand side of (\ref{BEC-eq})
\begin{equation}\label{I}
\mathcal{I}(\beta,\mu) = \lim_{\Lambda} \frac{1}{V} \sum_{k \in \Lambda^{*}} \omega_{\beta,\mu,\Lambda}^{0}(N_k)
= \frac{1}{(2\pi)^3}\int_{\mathbb{R}^3} d^3 k \ \frac{1}{e^{\beta \left(\varepsilon_{k}-\mu \right)}-1} \ ,
\end{equation}
exists for any $\mu <0$. It reaches its (finite) maximal value $\mathcal{I}(\beta,\mu =0) = \rho_c(\beta)$,
which is called the \textit{critical} particle density for a given temperature.

Recall that existence of the finite critical density $\rho_c(\beta)$ triggers (via the \textit{saturation}
mechanism) a \textit{zero-mode} {Bose-Einstein condensation} (BEC): $\rho_0(\beta) := \rho - \rho_c(\beta)$,
when the total particle density $\rho > \rho_c(\beta)$.
We note that indeed for $\alpha_1 < 1/2$, the totality of condensate $\rho_0(\beta)$ is sitting in the
one-particle ground  state mode $k=0$:
\begin{eqnarray} \label{BEC}
&&{\rho_0} (\beta)= {\rho} - \rho _{c}(\beta) = \lim_{\Lambda} \frac{1}{V}
\omega_{\beta,\mu_{\Lambda}(\beta,\rho),\Lambda}^{0}(b^{*}_{0} b_{0})
= \lim_{\Lambda} \frac{1}{V} \frac{1}{e^{-\beta \, {\mu_{\Lambda}(\beta,\rho\geq
\rho _{c}(\beta))} }-1} \ , \\
&&{\mu_{\Lambda}(\beta,\rho\geq \rho _{c}(\beta))}=  {- \, \frac{1}{V}} \ \frac{1}
{\beta(\rho -\rho _{c}(\beta))} + {o}({1}/{V}) \ , \\
&& \lim_{\Lambda} \frac{1}{V} \omega_{\beta,\mu,\Lambda}^{0}(b^{*}_{k} b_{k}) = 0 \ , \ \
\end{eqnarray}
where $\mu_{\Lambda}(\beta,\rho)$ is a unique solution of equation (\ref{BEC-eq}).

Following van den Berg-Lewis-Pul\'{e} \cite{vdBLP}) we introduce \textit{generalised} BEC (gBEC):
\begin{defi} \label{def:gBEC}
Total amount $\rho_{gBEC}(\beta, \mu)$ of the gBEC is defined by the
\textit{double-limit}:
\begin{equation} \label{gBEC}
\rho_{gBEC}(\beta, \mu) := \lim_{\delta \rightarrow +0}\lim_{\Lambda }
\frac{1}{V}\sum_{\left\{ k\in \Lambda^{\ast }, \,  \left\| k\right\|\leq \delta \right\}} \,
\omega_{\beta,\mu,\Lambda}(b^{*}_{k} b_{k}) \ .
\end{equation}
Here $\omega_{\beta,\mu,\Lambda}(\cdot)$ denotes the corresponding finite-volume grand-canonical Gibbs state.
\end{defi}

Then according to nomenclature proposed in \cite{vdBLP}) the {zero-mode} BEC in PBG is nothing but
the generalised Bose-Einstein condensation of the \textit{type} I. Indeed, by (\ref{BEC}) and (\ref{gBEC})
a non-vanishing BEC  \textit{implies} a nontrivial gBEC: $\rho_{0, gBEC}(\beta, \rho) > 0$. We denote this
{relation} as
\begin{equation}\label{BEC-gBEC}
{\rm{BEC}} \ \Rightarrow \ {\rm{gBEC}} \ .
\end{equation}
Moreover, (\ref{BEC}) and (\ref{gBEC}) yield: $\rho_{0}(\beta, \rho) = \rho_{0, gBEC}(\beta, \rho)$.

Recall that one also has BEC $\Rightarrow$ ODLRO, which is the \textit{Off-Diagonal-Long-Range-Order} for the
boson field
\begin{equation}\label{b-field}
b(x) = \sum_{k \in \Lambda^{*}} b_{k} {\varphi_{k}^{\Lambda}}(x) \ .
\end{equation}
Indeed, by definition of ODLRO \cite{Ver} the value of the off-diagonal \textit{spacial} correlation
$LRO(\beta,\rho)$ of the Bose-field is:
\begin{equation}\label{PBG-ODLRO}
LRO(\beta,\rho) = \lim_{\|x-y\|\rightarrow\infty}\lim_{\Lambda} \omega_{\beta,\mu_{\Lambda},\Lambda}^{0}
(b^*(x) \ b(y))=
\lim_{\Lambda} \omega_{\beta,\mu_{\Lambda},\Lambda}^{0}(\frac{b^{*}_{0}}{\sqrt{V}}\frac{b_{0}}{\sqrt{V}}) =
\rho_{0}(\beta,\rho) \ .
\end{equation}
Hence, (\ref{PBG-ODLRO}) coincides with the \textit{zero-mode} spacial \textit{averages} correlation of the
local observables (\ref{b-field}).

We recall that the $p$-mode \textit{spacial average} $\eta_{\Lambda,p}(b)$ of (\ref{b-field})
is equal to
\begin{equation}\label{b-p-mode-av}
\eta_{\Lambda,p}(b):= \frac{1}{V}\int_{\Lambda}dx \, b(x) \,  e^{- i \, p x} = \frac{b_p}{\sqrt{V}} \ ,
\ p \in \Lambda^{*} \ .
\end{equation}
As it is known, for PBG the value $LRO(\beta,\rho)$ of ODLRO {coincides} with the BEC (i.e., also with the
\textit{type} I gBEC) condensate density $\rho_{0}(\beta,\rho)$ \cite{vdBLP}.

To appreciate the relevance of gBEC versus quasi-averages we study more anisotropic thermodynamic limit:
$\alpha_1 = 1/2$,  known as the \textit{Casimir} box. Then one observes the infinitely-many state macroscopic
occupation, which is known as the gBEC of \textit{type} II defined by (\ref{gBEC}).
The total amount $\rho_{0}(\beta,\rho)$ of this condensate is asymptotically distributed between
infinitely-many low-energy microscopic states $\{\varphi_{k}^{\Lambda}\}_{k \in \Lambda^{*}}$ in such a
way that
\begin{eqnarray} \label{cond-II}
\rho_{0}(\beta,\rho)={\rho} - \rho _{c}(\beta) &=& \lim_{\delta \rightarrow +0}\lim_{\Lambda }
\frac{1}{V}\sum_{\left\{ k\in \Lambda^{\ast }, \,  \left\| k\right\|
\leq \delta \right\}}\left\{e^{\beta(\varepsilon_{k}- {\mu_{\Lambda}(\beta,\rho)})}- 1
\right\}^{-1} \\
&=& \sum_{n_1\in \mathbb{Z}}\frac{1}{(2\pi n_1)^2/2 + A} \ , \ \ {\rho} > \rho _{c}(\beta) \ . \nonumber
\end{eqnarray}
Here the parameter $A = A(\beta, \rho)\geq0$ is a {\textit{unique} root} of equation (\ref{cond-II}).
Then the amount of the zero-mode condensate BEC is:
\begin{eqnarray*}
\lim_{\Lambda} \frac{1}{V}\omega_{\beta,\mu,\Lambda}^{0}({b^{*}_{0} b_{0}}) = (A(\beta, \rho))^{-1}  \ .
\end{eqnarray*}
Note that in contrast to the case of \textit{type} I, the \textit{zero-mode} BEC $(A(\beta, \rho))^{-1}$ is
\textit{smaller} than gBEC of the \textit{type} II (\ref{cond-II}). Therefore, the relation between BEC and
gBEC is nontrivial.

To elucidate this point, we consider $\alpha_1 > 1/2$ (the \textit{van den Berg-Lewis-Pul\`{e}} box \cite{vdBLP}).
Then one obtains
\begin{equation}\label{BEC=0}
\lim_{\Lambda} \omega_{\beta,\mu,\Lambda}^{0}(\frac{b^{*}_{k} b_{k}}{V}) =
\lim_{\Lambda}\frac{1}{V}\left\{e^{\beta(\varepsilon_{k}-{\mu_{\Lambda}(\beta,\rho)})}-
1 \right\}^{-1} = 0  \ , \  \forall k \in \Lambda^{*} \ ,
\end{equation}
i.e., there is no macroscopic occupation of \textit{any} mode $k \in \Lambda^{*}$ for any value of particle
density $\rho$. So, density of the zero-mode BEC is zero, but the gBEC (called the \textit{type} III) does
exist in the same sense as it is defined by (\ref{gBEC}):
\begin{equation}\label{cond-III}
\rho -\rho_{c}(\beta)= \lim_{\delta \rightarrow +0}\lim_{\Lambda }
\frac{1}{V}\sum_{\left\{ k\in \Lambda^{\ast }, \left\| k\right\|
\leq \delta \right\}}\left\{e^{\beta(\varepsilon_{k}- {\mu_{\Lambda}(\beta,\rho)})}- 1
\right\}^{-1} > 0  , \ {\rm{for}} \ \ \rho > \rho_{c}(\beta) \ ,
\end{equation}
with the \textit{same} amount of the total density as that for types I and II.

We note that even for PBG the calculation of the ODLRO for the case of \textit{type} II and \textit{type} III
gBEC is a nontrivial problem.  This concerns, in particular, a regime when
there exists the \textit{second} critical density $\rho_{m}(\beta)>\rho_{c}(\beta)$ separating different
types of gBEC, see \cite{vdBLL} and \cite{BZ}. It is also clear that the zero-mode BEC is a more
\textit{restrictive} concept than the gBEC.

We comment that the fact that gBEC is different from BEC  is \textit{not} exclusively due to a special
anisotropy: $\alpha_1 > 1/2$,
or other geometries for the PBG, see \cite{BZ}. In fact the same phenomenon of the (\textit{type} III) gBEC
occurs due to repulsive \textit{interaction}. A simple example is the model with Hamiltonian \cite{ZBru}:
\begin{equation}\label{Int-TypeIII}
H_{\Lambda }= {\sum_{k\in \Lambda^{*}} }\varepsilon_{k}b_{k}^{*}b_{k}+
\frac{a}{2V}{\sum_{k\in\Lambda^{*}}} b_{k}^{*}b_{k}^{*}b_{k}b_{k}\ , \ \text{ } a>0 \ .
\end{equation}

Summarising we note that the concept of gBEC (\ref{gBEC}) covers the cases (e.g. (\ref{Int-TypeIII}))
when calculation of conventional BEC gives a trivial value: gBEC $\nRightarrow$ BEC, cf (\ref{BEC-gBEC}).
We also conclude that relations between BEC, gBEC and ODLRO are a subtle matter. This
motivates and bolsters a relevance of the Bogoliubov {quasi-average method} \cite{Bog07}-\cite{Bog70},
that we are going to consider also in connection with the Spontaneous Symmetry Breaking (SSB) of
gauge-invariance for the Gibbs states. We call the SSB of the gauge-invariance by the Gauge Symmetry Breaking
(GSB).

\subsection{Condensates, Bogoliubov quasi-averages and pure states} \label{sec:gBEC-BQ-A}
We now study the states of Boson systems, and for that matter assume (see \cite{Ver}, Ch.4.3.2), that they
are analytic in the sense of \cite{BR97}, Ch.5.2.3.

We start with the Hamiltonian for Bosons in a cubic box $\Lambda\subset \mathbb{R}^3$  of side $L$ with p.b.c.
and volume $V=L^{3}$:
\begin{equation}\label{4.1}
H_{\Lambda,\mu} = H_{0,\Lambda,\mu} + V_{\Lambda} \ ,
\end{equation}
where the interaction term has the form
\begin{equation}\label{4.2}
V_{\Lambda} = \frac{1}{2V} \sum_{{k},{p},{q}\in \Lambda^*}\,
\nu({p})b_{{k}+{p}}^{*}b_{{q}-{p}}^{*}b_{{q}} b_{{k}} \ ,
\end{equation}
Here $\nu$ is the Fourier transformation in $\mathbb{R}^3$ of the two-body potential
$v({x})$, with bound
\begin{equation}\label{4.3}
|\nu({k})| \le \nu({0}) < \infty \ .
\end{equation}


We define the group $G$ of (global) \textit{gauge} transformations $\{\tau_{s}\}_{s \in [0,2\pi)}$
by the Bogoliubov canonical mappings of CCR:
\begin{eqnarray}\label{4.17}
&& \tau_{s}(b^{*}(f)) = b^{*}( \exp(i \, s) f) = \exp(i \, s)b^{*}(f) \ , \\
&& \tau_{s}(b(f)) = b( \exp(i \, s) f) = \exp(-i \, s)b(f) \ , \nonumber
\end{eqnarray}
where $b^{*}(f)$ and $b(f)$ are the creation and annihilation operators smeared over test-functions $f$
from the Schwartz space. Note that for $f=\varphi_{k}^{\Lambda}$ they coincide with $b^{*}_{k},  b_{k}$, cf
(\ref{b-field}), and $\tau_{s} (\cdot) = \exp(i \, s N_{\Lambda}) (\cdot) \exp(- i \, s N_{\Lambda})$,
see (\ref{G-C-PBG}). By definition (\ref{4.17}) and by virtue of (\ref{G-C-PBG}), (\ref{4.2})
the Hamiltonian (\ref{4.1}) is gauge-invariant:
\begin{equation}\label{4.31}
H_{\Lambda,\mu} = e^{i \, s N_{\Lambda}} H_{\Lambda,\mu} e^{- i \, s N_{\Lambda}} \ .
\end{equation}
Note that the property (\ref{4.31}) evidently implies the \textit{gauge-invariance} of the Gibbs state
(\ref{0-state}) as well as that for Hamiltonian (\ref{4.1}) of imperfect Bose-gas:
\begin{equation}\label{4.32}
\omega_{\beta,\mu,\Lambda}(\cdot) = \omega_{\beta,\mu,\Lambda}(\tau_{s}(\cdot)) =
\frac{{\rm{Tr}}_{{\mathcal{F}}_{\Lambda}}
(\exp(-\beta H_{\Lambda,\mu})\ \tau_{s}(\cdot) \ )}
{{\rm{Tr}}_{{\mathcal{F}}_{\Lambda}} \exp(-\beta H_{\Lambda,\mu})} \ \ .
\end{equation}
Symmetry (\ref{4.32}) is a source of \textit{selection} rules. For example:
\begin{equation}\label{4.33}
\omega_{\beta,\mu,\Lambda}(A_{n,m}) = 0 \ , \ {\rm{for}} \ \
A_{n,m}=\prod_{i=1,\, j=1}^{n,\, m} b^{*}_{k_i} b_{k_j} \ , \ \ {\rm{if}} \ \ n \neq m \ .
\end{equation}

The \textit{quasi}-Hamiltonian corresponding to (\ref{4.1}) with  \textit{gauge symmetry} breaking
sources is taken to be
\begin{equation}\label{4.6}
H_{\Lambda,\mu,\lambda_{\phi}} = H_{\Lambda,\mu} + H_{\Lambda}^{\lambda_{\phi}} \ .
\end{equation}
Here the sources are switched on only in zero mode ($k=0$):
\begin{equation}\label{4.7}
H_{\Lambda}^{\lambda_{\phi}} = \sqrt{V}(\bar{\lambda}_{\phi} b_{{0}}+\lambda_{\phi} b_{{0}}^{*}) \ ,
\end{equation}
for
\begin{equation}\label{4.8}
\lambda_{\phi} = \lambda \exp(i\phi) \ \mbox{ with } \lambda \geq 0 \, , \,
\mbox{ where } \, {\rm{arg}}(\lambda_{\phi}) = \phi \in [0,2\pi) \ .
\end{equation}
In this case the corresponding Gibbs state is \textit{not} gauge-invariant (\ref{4.33}) since, for example
\begin{equation}\label{4.34}
\omega_{\beta,\mu,\Lambda,\lambda_{\phi}}(b_{{k}}) =
\frac{{\rm{Tr}}_{{\mathcal{F}}_{\Lambda}}
(\exp(-\beta H_{\Lambda,\mu,\lambda_{\phi}})\ b_{{k}} \ )}
{{\rm{Tr}}_{{\mathcal{F}}_{\Lambda}} \exp(-\beta H_{\Lambda,\mu,\lambda_{\phi}})} \neq 0
\ \ {\rm{for}} \ \ k=0 \ .
\end{equation}
The GSB of the state (\ref{4.34}), which is induced by the sources in (\ref{4.6}), persists in the thermodynamic
limit for the state $\omega_{\beta,\mu,\lambda_{\phi}}(\cdot):=
\lim_{V \rightarrow \infty}\omega_{\beta,\mu,\Lambda,\lambda_{\phi}}(\cdot)$. But it may occur in this
limit spontaneously without external sources. Let us denote
\begin{equation}\label{4.431}
\omega_{\beta,\mu}(\cdot):=
\lim_{V \rightarrow \infty}\omega_{\beta,\mu,\Lambda,\lambda_{\phi}=0}(\cdot) \ .
\end{equation}
\begin{defi} \label{SSB}
We say that the state $\omega_{\beta,\mu}$ undergoes a spontaneous breaking
of the $G$-invariance (\textit{spontaneous} Gauge Symmetry Breaking (GSB)), if:\\
(i) $\omega_{\beta,\mu}$ is $G$-invariant,\\
(ii) $\omega_{\beta,\mu}$ has a nontrivial decomposition into ergodic states $\omega_{\beta,\mu}^{'}$,
which means that at least two such distinct states occur in representation
\begin{equation*}
\omega_{\beta,\mu}(\cdot) = \int_{0}^{2\pi} d\nu(s) \ \omega_{\beta,\mu}^{'}(\tau_{s} \ \cdot)\ ,
\end{equation*}
and for some $s$
\begin{equation*}
\omega_{\beta,\mu}^{'}(\tau_{s}\cdot) \ne \omega_{\beta,\mu}^{'}(\cdot) \ .
\end{equation*}
Note that ergodic states are characterized by the \textit{clustering} property,
which implies a decorrelation of the zero-mode spacial averages (\ref{b-p-mode-av}) for the PBG,
as well as in general for the imperfect Bose gas.
\end{defi}
We take initially $\lambda \geq 0$ and consider first the perfect Bose-gas (\ref{G-C-PBG}) to define
the Hamiltonian
\begin{equation}\label{4.9.1}
H_{0,\Lambda, \mu, \lambda_{\phi}} = H_{0,\Lambda,\mu} + H_{\Lambda}^{\lambda_{\phi}} \ ,
\end{equation}
which is \textit{not} globally gauge-invariant. To separate the symmetry-breaking term $H_{{0}}$ we rewrite
(\ref{4.9.1}) as
\begin{equation*}
H_{0,\Lambda,\mu,\lambda_{\phi}}= H_{{0}}+H_{{k}\ne{0}} \ ,
\end{equation*}
where $H_{{0}} = -\mu \ b_{{0}}^{*}b_{{0}}+\sqrt{V}(\bar{\lambda}_{\phi} b_{{0}}+
\lambda_{\phi} b_{{0}}^{*}) = -\mu (b_{{0}} - \sqrt{V}{\lambda}_{\phi}/\mu)^{*}
(b_{{0}} - \sqrt{V}{\lambda}_{\phi}/\mu) +  V |{\lambda}_{\phi}|^2/\mu $.

Recall that for the perfect Bose-gas the grand-canonical partition function $\Xi_{0,\Lambda}$ splits into
a product over the zero mode and the remaining modes. We introduce the canonical shift transformation
\begin{equation}\label{4.9.2}
\widehat{b}_{{0}} := b_{{0}} - \frac{\lambda_{\phi} \sqrt{V}}{\mu} \ ,
\end{equation}
without altering the nonzero modes. Since $\mu < 0$, we thus obtain for the grand-canonical partition
function $\Xi_{0,\Lambda}$,
\begin{equation}\label{4.9.3}
\Xi_{0,\Lambda}(\beta,\mu,\lambda_{\phi}) = (1-\exp(\beta \mu))^{-1}
\exp(-\frac{\beta |\lambda_{\phi}|^{2}}{\mu} V)\ \Xi^{\prime}_{0,\Lambda}(\beta,\mu) \ ,
\end{equation}
where
\begin{equation}\label{4.9.4}
\Xi^{\prime}_{0,\Lambda}(\beta,\mu) := \prod_{{k} \ne {0}}  (1-\exp(-\beta(\epsilon_{{k}}-\mu)))^{-1} \ ,
\end{equation}
with $\epsilon_{{k}}={k}^{2}$. Recall that the grand-canonical state for the perfect Bose-gas is
\begin{equation}\label{4.9.5}
\omega^{0}_{\beta,\mu,\Lambda,\lambda_{\phi}}(\cdot):= {\frac{1}{\Xi_{0,\Lambda}(\beta,\mu,\lambda_{\phi})}} \
{\rm{Tr}}_{{\mathcal{F}}_{\Lambda}}[e^{-\beta H_{0,\Lambda,\mu,\lambda_{\phi}}} \ (\cdot )]  \ .
\end{equation}
see Section \ref{sec:gBEC}. Then it follows from (\ref{4.9.3})-(\ref{4.9.5}) that the mean density ${\rho}$
equals to
\begin{equation}\label{4.9.6}
{\rho}=\omega^{0}_{\beta,\mu,\Lambda,\lambda_{\phi}}(\frac{N_{\Lambda}}{V})=
\frac{|\lambda_{\phi}|^{2}}{\mu^{2}} + \frac{1}{V}\frac{1}{\exp(-\beta \mu)-1} +
\frac{1}{V} \sum_{{k} \ne {0}} \frac{1}{\exp(\beta(\epsilon_{{k}}-\mu))-1} \ .
\end{equation}

Equation (\ref{4.9.6}) is the starting point of our analysis. Since the critical density
$\rho_{c}(\beta)= \mathcal{I}(\beta,\mu =0)$ is finite (\ref{I}), we have the following statement.
\begin{proposition}\label{prop:4.1}
Let $0 < \beta <\infty$ be fixed. Then, for each
\begin{equation}\label{4.11}
{\rho_{c}}(\beta) < {\rho} < \infty \ ,
\end{equation}
and for each $\lambda >0$, $V <\infty$, there exists a unique solution of (\ref{4.9.6}) of the form
\begin{equation}
\mu_{\Lambda}({\rho},|\lambda_{\phi}|) = -\frac{|\lambda_{\phi}|}{\sqrt{{\rho}-{\rho_{c}}(\beta)}}\\
 + \alpha(|\lambda_{\phi}|,V) \ ,
\label{4.12.1}
\end{equation}
with
\begin{equation}\label{4.12.2}
\alpha(|\lambda_{\phi}|,V) \ge 0 \ \ \forall \ |\lambda_{\phi}|, V \ ,
\end{equation}
and such that
\begin{equation}\label{4.13}
\lim_{|\lambda_{\phi}| \to 0} \lim_{V \to \infty} \frac{\alpha(|\lambda_{\phi}|,V)}{|\lambda_{\phi}|} = 0 \ .
\end{equation}
\end{proposition}
\begin{rem}\label{rem:4.2} The proof of this statement is straightforward and follows from equation
(\ref{4.9.6}). {We also note that besides the cube $\Lambda$, the Proposition \ref{prop:4.1} is also true
for the case of three-dimensional anisotropic parallelepiped
$\Lambda:= V^{\alpha_1}\times V^{\alpha_2}\times V^{\alpha_3}$, with p.b.c. and
$\alpha_1 \geq \alpha_2 \geq \alpha_3$, $\alpha_1 + \alpha_2 + \alpha_3 = 1$, i.e. when for $\lambda =0$
one has \textit{type} II or \textit{type} III condensations .}
\end{rem}
Since $|\lambda_{\phi}|^{2} = \lambda_{\phi}\bar{\lambda}_{\phi} = \lambda^{2}$, we obtain that the limit of
expectation
\begin{equation}\label{4.14.1}
\lim_{\lambda \to +0} \lim_{V \to \infty}\omega^{0}_{\beta,\mu,\Lambda,\lambda_{\phi}}(b_{{0}}^{*}/\sqrt{V})=
- \lim_{\lambda \to +0} \lim_{V \to \infty} \frac{\partial}{\partial \lambda_{\phi}}
p_{\beta,\mu,\Lambda,\lambda_{\phi}} \ ,
\end{equation}
is related to derivative of the \textit{grand-canonical pressure} with respect to the breaking-symmetry sources
(\ref{4.6}):
\begin{equation}\label{4.14.2}
 p_{\beta,\mu,\Lambda,\lambda_{\phi}}:=\frac{1}{\beta V}\ln \Xi_{0,\Lambda}(\beta,\mu,\lambda_{\phi}) \ .
\end{equation}
Recall that the left-hand side of (\ref{4.14.1}) is in fact the Bogoliubov quasi-average of
$b_{{0}}^{*}/\sqrt{V}$.

By (\ref{4.9.4}) and (\ref{4.14.2})
we obtain that
\begin{equation}\label{4.14.3}
\frac{\partial}{\partial \lambda_{\phi}} p_{\beta,\mu,\Lambda,\lambda_{\phi}}=
-\frac{\bar{\lambda}_{\phi}}{\mu} \ .
\end{equation}
Since for a given $\rho$ the asymptotic of the chemical potential is (\ref{4.12.1}),
by (\ref{4.14.1}) and (\ref{4.14.3}) one gets
\begin{equation}\label{4.15.1}
\lim_{\lambda \to +0} \lim_{V \to \infty}\omega^{0}_{\beta,\mu,\Lambda,\lambda_{\phi}}(b_{{0}}^{*}/\sqrt{V})
= \sqrt{\rho_{{0}}(\beta,\rho)} \exp(- i\phi) \ ,
\end{equation}
where according to (\ref{4.9.6}) and (\ref{4.14.3})
\begin{equation*}
\rho_{{0}}(\beta,\rho) = {\rho}-{\rho_{c}}(\beta, \rho) \ ,
\end{equation*}
is the perfect Bose-gas condensation in \textit{zero mode}. We see therefore that the phase in (\ref{4.14.1})
remains in (\ref{4.15.1}) even after the limit $\lambda \to +0$.

In \cite{LSYng} the following definition of $\rm{(GSB)}_{q-a}$ was suggested in the more general framework
of the imperfect Bose gas that we consider later:
\begin{defi} \label{SSBq-a}
We say that the state $\omega_{\beta,\mu,\Lambda,\lambda_{\phi}}$ undergoes a \textit{spontaneous} Gauge
Symmetry Breaking $\rm{(GSB)}_{q-a}$ in the Bogoliubov q-a sense if the limit state (\ref{4.431})
rests gauge-invariant, whereas the state
\begin{equation}\label{GSB-q-a}
\omega_{\beta,\mu,\phi} (\cdot) := \lim_{\lambda \to +0}
\lim_{V \to \infty}\omega_{\beta,\mu,\Lambda,\lambda_{\phi}} (\cdot) \ ,
\end{equation}
is not gauge-invariant and $\omega_{\beta,\mu,\phi} \neq \omega_{\beta,\mu,\phi^{\prime}}$, when
$\phi \neq \phi^{\prime}$.
\end{defi}
We note that $\rm{(GSB)}_{q-a}$ is {equivalent} to $\rm{(GSB)}$, i.e. to
Definition \ref{SSB}, where the ergodic states $\omega_{\beta,\mu}^{'}$ in (ii) coincide with the set
of $\omega_{\beta,\mu,\phi}$ in (\ref{GSB-q-a}), see Theorem \ref{theo:3.4} below.
The notion $\rm{(GSB)}_{q-a}$ is, however, useful for purposes of comparison with \cite{LSYng}.
\begin{rem}\label{rem:4.3}  Note that by (\ref{4.9.6}) together with Proposition \ref{prop:4.1} and
(\ref{4.15.1}) one gets
\begin{eqnarray}\label{4.15.11}
&&\rho_{{0}}(\beta,\rho) = \lim_{\lambda \to +0} \lim_{V \to \infty}
\omega^{0}_{\beta,\mu,\Lambda,\lambda_{\phi}}(\frac{b_{{0}}^{*}}{\sqrt{V}}\frac{b_{{0}}}{\sqrt{V}}) = \\
&&\lim_{\lambda \to +0} \lim_{V \to \infty}\omega^{0}_{\beta,\mu,\Lambda,\lambda_{\phi}}(b_{{0}}^{*}/\sqrt{V})
\lim_{\lambda \to +0} \lim_{V \to \infty}\omega^{0}_{\beta,\mu,\Lambda,\lambda_{\phi}}(b_{{0}}/\sqrt{V}) \ .
\nonumber
\end{eqnarray}
Besides \textit{decorrelation} of the zero-mode \textit{spacial} averages
$\eta_{\Lambda, 0}(b^*) = b_{{0}}^{*}/\sqrt{V}$ and $\eta_{\Lambda, 0}(b) = b_{{0}}/\sqrt{V}$,
(\ref{b-p-mode-av}), for the Bogoliubov q-a, equation (\ref{4.15.11}) establishes \textit{also}
the identity between zero-mode condensation fraction $\rho_{{0}}(\beta,\rho)$ and
$LRO(\beta,\rho)$ (\ref{PBG-ODLRO}), that we denote by $\rm{(ODLRO)}_{q-a}$.
Decorrelation in the right-hand side of (\ref{4.15.11}) indicates for the Bogoliubov q-a \textit{a nontrivial}
$\rm{(GSB)}_{q-a}$ in the presence of condensate, see (\ref{4.15.1}) and Definition \ref{SSBq-a}.
\end{rem}
Remarks \ref{rem:4.2} and \ref{rem:4.3} motivate definition of the q-a \textit{states} for the perfect Bose-gas
as follows:
\begin{equation}\label{PBG-LimSt-phi}
\omega^{0}_{\beta,\mu,\phi} := \lim_{\lambda \to +0}
\lim_{V \to \infty}\omega^{0}_{\beta,\mu,\Lambda,\lambda_{\phi}} \ ,
\end{equation}
where the double limit along a subnet $\Lambda \uparrow \mathbb{R}^3$ exists by weak* compactness of the
set of states \cite{BR87}.
Below we use notation $\omega$ for the Gibbs state in general case (\ref{4.1})-(\ref{4.3}) and we keep
$\omega^{0}$ for the perfect Bose-gas.
\begin{defi}\label{defi:2.1}
We recall that Bose-gas undergoes the \textit{zero}-mode BEC if
\begin{equation}\label{4.16}
\lim_{V \to \infty} \frac{1}{V}\omega_{\beta,\mu,\Lambda}({b_{{0}}^{*}b_{{0}}}) =
\lim_{V \to \infty} \frac{1}{V^2} \int_{\Lambda} \int_{\Lambda} dx \, dy \,
\omega_{\beta,\mu,\Lambda}(b^*(x)b(y)) > 0 \ .
\end{equation}
Simultaneously, this means a non-trivial correlation (\ref{PBG-ODLRO})
\begin{equation}\label{4.161}
\lim_{\|x-y\| \to \infty} \lim_{V \to \infty} \omega_{\beta,\mu,\Lambda}(b^*(x)b(y)) > 0 \ ,
\end{equation}
of zero-mode spacial averages (\ref{b-p-mode-av}), that we denoted by ODLRO .
\end{defi}
As we demonstrated in Section \ref{sec:gBEC} even for the PBG this definition is too restricted since
(\ref{4.16}) might be trivial, although condensation does exist because of a finite critical density
$\rho_{c}(\beta,\mu)$.
We say that Bose-gas undergoes gBEC (Definition \ref{def:gBEC}) if
\begin{equation}\label{4.17}
\lim_{\delta \rightarrow +0}\lim_{\Lambda } \frac{1}{V}\sum_{\left\{ k\in \Lambda^{\ast }, \,
\left\| k\right\| \leq \delta \right\}} \omega_{\beta,\mu,\Lambda}({b_{{k}}^{*}b_{{k}}}) =
\rho -\rho_{c}(\beta,\mu) > 0 \ .
\end{equation}

To classify different \textit{types} of the gBEC one has to consider the value of the limits:
\begin{equation}\label{4.18}
\lim_{\Lambda } \frac{1}{V}\ \omega_{\beta,\mu,\Lambda}({b_{{k}}^{*}b_{{k}}}) =: \rho_k \ , \
k\in \Lambda^{\ast } \ .
\end{equation}
Then according to Section \ref{sec:gBEC}, one has $\rho_{k=0} = \rho -\rho_{c}$
for the type I gBEC, $\rho_{k=0} < \rho -\rho_{c}$ for the type II gBEC. If one has
$\{\rho_k = 0\}_{k\in \Lambda^{\ast }}$ and non-trivial (\ref{4.17}), then the gBEC is of the type III.
\begin{defi}\label{defi:2.3}
We say that Bose-gas undergoes Bogoliubov quasi-average {condensation} $\rm{(BEC)}_{q-a}$ if
\begin{equation}\label{4.19}
\lim_{\lambda \to +0} \lim_{V \to \infty}
\omega_{\beta,\mu,\Lambda,\lambda_{\phi}}(\frac{b_{{0}}^{*}}{\sqrt{V}}\frac{b_{{0}}}{\sqrt{V}}) > 0 \ .
\end{equation}
\end{defi}
\begin{rem}\label{rem:3.3}
First, the results of Remark \ref{rem:4.3} are \textit{independent} of the anisotropy, i.e. of whether the
condensation for $\lambda =0$ is in single mode ($k=0$) (i.e. BEC) or it is extended as the gBEC-type III,
Section \ref{sec:gBEC}. We comment that the condensate in the mode $k=0$ is due to the one-particle Hamiltonian
\textit{spectral property} that implies $\varepsilon_{k=0}=0$ (\ref{dual-Lambda}).

Second, these results yield that the Bogoliubov quasi-average method solves for PBG the
question about \textit{equivalence} between $\rm{(BEC)}_{q-a}$, $\rm{(GSB)}_{q-a}$ and $\rm{(ODLRO)}_{q-a}$:
\begin{equation}\label{PBG-qa-equiv}
\rm{(BEC)}_{q-a} \Leftrightarrow \rm{(ODLRO)}_{q-a} \Leftrightarrow \rm{(GSB)}_{q-a} \ ,
\end{equation}
which holds if they are defined via the \textit{one-mode} quasi-average for $k=0$.
Here equivalence $\Leftrightarrow$ means implications in both sense.
\end{rem}
Then the quasi-average for $k\neq 0$, i.e. for $\varepsilon_{k} > 0$, needs a certain elucidation.
To this aim we revisit the prefect Bose-gas (\ref{G-C-PBG}) with symmetry breaking sources (\ref{4.7})
in a single mode $q \in \Lambda^{*}$, which is in general not a zero-mode:
\begin{eqnarray}\label{freeQE}
H^{0}_{\Lambda} (\mu; h) \, := \, H^{0}_{\Lambda}(\mu) \, + \, \sqrt{V} \ \big( \overline{h} \  b_{{q}} +
h \ b^{*}_{{q}} \big) \ , \ \mu \leq 0.
\end{eqnarray}
Then for a fixed density ${\rho}$, the the grand-canonical condensate equation (\ref{BEC-eq}) for (\ref{freeQE})
takes the following form:
\begin{eqnarray}\label{perfect-gas-with-source-density-equation-finite-volume}
&&{\rho} = \rho_{\Lambda}(\beta, \mu, h) \, := \, \frac{1}{V} \sum_{k \in \Lambda^{*}_{l}}
\omega_{\beta,\mu,\Lambda,h}^{0}(b^{*}_{k}b_{k}) = \\
&&\frac{1}{V} (e^{\beta(\varepsilon_{{q}} - \mu)}-1)^{-1} \, + \, \frac{1}{V}
\sum_{k\in \Lambda^{*}\setminus{q}} \frac{1}{e^{\beta(\varepsilon_{k} - \mu)}-1} \, + \,
\frac{\vert h \vert\, ^{2}}
{(\varepsilon_{{q}} - \mu)\, ^{2}} \ . \nonumber
\end{eqnarray}

According the quasi-average method, to investigate a possible condensation, one must first take the thermodynamic
limit in the right-hand side of (\ref{perfect-gas-with-source-density-equation-finite-volume}), and then switch
off the symmetry breaking source: $h \rightarrow 0$. Recall that the critical density, which defines the
threshold of boson saturation is equal to $\rho_c(\beta) = \mathcal{I}(\beta,\mu=0)$ (\ref{I}), where
$\mathcal{I}(\beta,\mu)=\lim_{\Lambda} \rho_{\Lambda}(\beta, \mu , h = 0)$.

Since $\mu \leq 0$, we now have to distinguish two cases:\\
(i) Let the mode ${q}\in \Lambda^{*}$ be such that $\lim_{\Lambda} \varepsilon_{{q}} > 0$. Then we obtain from
(\ref{perfect-gas-with-source-density-equation-finite-volume}) for the condensate equation and for the simplest
$q$-mode gauge-symmetry breaking expectation:
\begin{eqnarray*}
{\rho} \, = \, \lim_{h \rightarrow 0}
\lim_{\Lambda} \rho_{\Lambda}(\beta, \mu, h) \, = \, \mathcal{I}(\beta, \mu) \ , \  \
\lim_{h \to 0} \lim_{V \to \infty}
\omega_{\beta,\mu,\Lambda,h}^{0}(\frac{b_{{q}}^{*}}{\sqrt{V}}) =
\lim_{h \to 0} \frac{ \overline{h} } {(\varepsilon_{{q}} - \mu)\,} = 0
 \ .
\end{eqnarray*}
This means that the quasi-average coincides with the average. Hence, we return to the analysis of the condensate
equation (\ref{perfect-gas-with-source-density-equation-finite-volume}) for $h =0$. This leads to finite-volume
solutions $\mu_{\Lambda}(\beta,\rho)$ and consequently to all possible types of condensation as a function of
anisotropy $\alpha_1$, see Section \ref{sec:gBEC} for details.\\
(ii) On the other hand, if ${q}\in \Lambda^{*}$ is such that $\lim_{\Lambda} \varepsilon_{{q}} = 0$, then
thermodynamic limit in the right-hand side of the condensate equation
(\ref{perfect-gas-with-source-density-equation-finite-volume}) and the
$q$-mode gauge-symmetry breaking expectation yield:
\begin{eqnarray}\label{perfect-gas-with-source-density-equation-infinite-volume}
{\rho} =  \lim_{\Lambda} \rho_{\Lambda}(\beta, \mu, h)
\, = \, \mathcal{I}(\beta, \mu) + \frac{\vert h \vert\, ^{2}}{\mu\, ^{2}} \ , \  \
\lim_{V \to \infty}
\omega_{\beta,\mu,\Lambda,h}^{0}(\frac{b_{{q}}^{*}}{\sqrt{V}}) =
\frac{ \overline{h} } {(- \mu)\,}  \ .
\end{eqnarray}

If ${\rho} \leq \rho_{c}(\beta)$, then the limit of solution of
(\ref{perfect-gas-with-source-density-equation-infinite-volume}):
$\lim_{h \rightarrow 0}{\mu}(\beta, {\rho}, h) = {\mu}_{0} (\beta, {\rho}) <0$,
where ${\mu}(\beta,{\rho}, h)= \lim_{\Lambda}{\mu}_{\Lambda} (\beta,{\rho}, h)<0 $ is thermodynamic limit of
the finite-volume solution of condensate equation (\ref{perfect-gas-with-source-density-equation-finite-volume}).
Therefore, there is no condensation in any mode and according to
(\ref{perfect-gas-with-source-density-equation-infinite-volume}) the corresponding  $q$-mode gauge-symmetry
breaking expectation for $h \to 0$ (Bogoliubov quasi-average)  again equals to zero.

But if ${\rho} > \rho_{c}(\beta)$, then (\ref{perfect-gas-with-source-density-equation-finite-volume})
yields that $\lim_{h \rightarrow 0}{\mu}(\beta, {\rho},h) =0$. Therefore, by
(\ref{perfect-gas-with-source-density-equation-infinite-volume}) the density of condensate and the
Bogoliubov quasi-average are
\begin{eqnarray}\label{BEC-qa}
&& \rho_{0}(\beta) = {\rho} - \rho_{c}(\beta) =
\lim_{h \rightarrow 0}\frac{\vert h \vert\, ^{2}}{\mu(\beta, {\rho},h)\, ^{2}} \ \ , \\
&& \lim_{h \rightarrow 0} \lim_{V \to \infty} \omega_{\beta,{\mu}_{\Lambda}(\beta,{\rho}, h),\Lambda,h}^{0}
(\frac{b_{{q}}^{*}}{\sqrt{V}}) =
\lim_{h \rightarrow 0} \lim_{V \to \infty}
\omega_{\beta,{\mu}_{\Lambda}(\beta,{\rho}, h),\Lambda,h}^{0}(\frac{b_{{0}}^{*}}{\sqrt{V}}) =
\sqrt{\rho_{0}(\beta)} e^{- i \, {\rm{arg}} (h)}  \nonumber \ .
\end{eqnarray}

Consider now the \textit{case} (i) in more details. Let $\lim_{\Lambda} \varepsilon_{{q}} =: \varepsilon_{{q}} > 0$.
Then by (\ref{perfect-gas-with-source-density-equation-finite-volume})
for the finite-volume expectation of the particle density in the $q$-mode is
\begin{equation}\label{BEC-q-posit}
\omega_{\beta,\mu,\Lambda,h}^{0}({b^{*}_{q}b_{q}}/{V}) = \frac{1}{V} (e^{\beta(\varepsilon_{{q}} - \mu)}-1)^{-1}
+ \frac{\vert h \vert\, ^{2}} {(\varepsilon_{{q}} - \mu)\, ^{2}} \ .
\end{equation}
Since the one-particle spectrum $\{\varepsilon_{{k}}\geq 0\}_{k\in\Lambda^*}$ and $\varepsilon_{{k=0}} = 0$
(\ref{dual-Lambda}), the solution of equation (\ref{perfect-gas-with-source-density-equation-finite-volume})
is unique and negative: ${\mu}_{\Lambda} (\beta,{\rho}, h)<0$.
Then the Bogoliubov quasi-average of ${b^{*}_{q}b_{q}}/{V}$ is equal to
\begin{eqnarray}\label{Bog-qa}
&&\lim_{h \rightarrow 0}\lim_{\Lambda}\omega_{\beta,{\mu}_{\Lambda}
(\beta,{\rho}, h),\Lambda,h}^{0}({b^{*}_{q}b_{q}}/{V}) =  \\
&&\lim_{h \rightarrow 0}\lim_{\Lambda}\frac{1}{V} (e^{\beta(\varepsilon_{{q}} - {\mu}_{\Lambda}
(\beta,{\rho}, h))}-1)^{-1} +
\lim_{h \rightarrow 0}\lim_{\Lambda} \frac{\vert h \vert\, ^{2}} {(\varepsilon_{{q}} -
{\mu}_{\Lambda} (\beta,{\rho}, h))\, ^{2}}  = 0 \ , \nonumber
\end{eqnarray}
for any particle density including the case ${\rho} > \rho_{c}(\beta)$.

Now the condensate equation (\ref{perfect-gas-with-source-density-equation-infinite-volume}) and the
$q$-mode gauge-symmetry breaking expectation get the form:
\begin{eqnarray}\label{perfect-gas-with-source-density-equation-infinite-volume-q}
&&{\rho} =  \lim_{\Lambda} \rho_{\Lambda}(\beta, \mu, h)
\, = \, \mathcal{I}(\beta, \mu) + \frac{\vert h \vert\, ^{2}}{(\varepsilon_{{q}} - \mu )^{2}}
=: \rho(\beta, \mu, h) \ , \\
&&\lim_{V \to \infty}
\omega_{\beta,\mu,\Lambda,h}^{0}(\frac{b_{{q}}^{*}}{\sqrt{V}}) =
\frac{ \overline{h} } {(\varepsilon_{{q}} - \mu)\,}  \label{infinite-volume-q} \ .
\end{eqnarray}
\begin{rem}\label{rem:3.31}
Note that (\ref{freeQE}) gives an example of the model of condensation that
depend on \textit{external source} in non-zero mode. Indeed, for the perfect Bose-gas with the one-particle
spectrum (\ref{dual-Lambda})  the solution ${\mu}(\beta,{\rho}, h))$ of the
condensate equation (\ref{perfect-gas-with-source-density-equation-infinite-volume-q}) is such that
\begin{equation*}
\lim_{\rho \rightarrow \rho_{c}(\beta, h)}{\mu}(\beta,{\rho}, h)) = 0 \ \ \ {\rm{and}} \ \ \
\rho_{c}(\beta, h) := \sup_{\mu \leq 0} \rho(\beta, \mu, h) = \rho(\beta, \mu =0, h) \ .
\end{equation*}
Since $\varepsilon_{{q}} > 0$ and $\varepsilon_{{0}} = 0$ the finite saturation density $\rho_{c}(\beta, h)$
trigger BEC in the zero mode of perfect Bose-gas (\ref{freeQE}) if $\rho > \rho_{c}(\beta, h)$. To this end
we observe that by (\ref{perfect-gas-with-source-density-equation-finite-volume}), (\ref{BEC-q-posit}) and
(\ref{perfect-gas-with-source-density-equation-infinite-volume-q}) one finds
\begin{eqnarray}\label{BEC-q}
{\rho} - \rho_{c}(\beta, h) =
\lim_{\Lambda} \frac{1}{V} \omega_{\beta,{\mu}_{\Lambda} (\beta,{\rho}, h),\Lambda,h}^{0}(b^{*}_{0}b_{0}) \ ,
\end{eqnarray}
where solution of equation (\ref{perfect-gas-with-source-density-equation-finite-volume}) has for
$V \rightarrow \infty$ the asymptotics:
\begin{equation*}
{\mu}_{\Lambda} (\beta,{\rho}, h) = - ({\rho} - \rho_{c}(\beta, h))V^{-1} + {o} (V^{-1}) \ .
\end{equation*}
Therefore, the model (\ref{freeQE}) is the ideal Bose-gas with external sources, which behaviour
is almost identical to Bose-gas with $h=0$, Section \ref{sec:gBEC}. This concerns the higher critical density:
$\rho(\beta, \mu =0, h) \geq \rho_{c}(\beta)$ (\ref{perfect-gas-with-source-density-equation-infinite-volume-q})
and non-trivial expectation of the particle density (\ref{BEC-q-posit}) in a non-zero $q$-mode.
\end{rem}

\textit{Summarising the case} (i). The non-zero mode sources for the ideal Bose-gas and the corresponding
Bogoliubov quasi-averages give the \textit{same} results as for the ideal Bose-gas \textit{without} external
sources. Hence, the quasi-averages in this case have \textit{no impact} and lead to the same conclusions
(and problems) as the generalised  BEC in Section \ref{sec:gBEC}. If one keeps the non-zero mode source, then
this generalised BEC has a \textit{source-dependent} critical density as in Remark \ref{rem:3.31}.

\textit{Summarising the case} (ii). First we note that by virtue of
(\ref{perfect-gas-with-source-density-equation-finite-volume}),
(\ref{perfect-gas-with-source-density-equation-infinite-volume}) one has
${\mu}(\beta,{\rho}, h \neq 0)<0$ and that for any $k \neq q \, $, even when
$\lim_{\Lambda} \varepsilon_{{k}} = 0 \, $,
\begin{equation}\label{zero-non-zero-modes}
\lim_{h \rightarrow 0}\lim_{\Lambda}\omega_{\beta,{\mu}_{\Lambda}
(\beta,{\rho}, h),\Lambda,h}^{0}({b^{*}_{k}b_{k}}/{V}) =
\lim_{h \rightarrow 0}\lim_{\Lambda} \frac{1}{V}
\frac{1}{e^{\beta(\varepsilon_{{k}}- {\mu}_{\Lambda} (\beta,{\rho}, h)))}-1} = 0 \ .
\end{equation}
This means for any anisotropy $\alpha_1$ the \textit{quasi-average} condensation $\rm{(BEC)}_{q-a}$ occurs
only in one zero-mode (BEC type I), whereas the gBEC for $\alpha_1 >1/2$ is of the type III,
see Section \ref{sec:gBEC}.
Diagonalisation (\ref{4.9.2}) for $b_{{q}} \rightarrow \widehat{b}_{{q}} $, and (\ref{BEC-qa}) allow to
apply the quasi-average method to calculate a nonvanishing for ${\rho} > \rho_{c}(\beta)$ gauge-symmetry
breaking $\rm{(GSB)}_{q-a}$:
\begin{equation}\label{GSB-qa}
\lim_{h \rightarrow 0}\lim_{\Lambda}\omega_{\beta,{\mu}_{\Lambda}
(\beta,{\rho}, h),\Lambda,h}^{0}({b_{q}}/\sqrt{V}) =
\lim_{h \rightarrow 0}\frac{h}{\mu(\beta,{\rho}, h)} =
 e^{i \, {\rm{arg}}(h)} \, \sqrt{{\rho} - \rho_{c}(\beta)} \ ,
\end{equation}
along $\{h = |h| e^{i \, {\rm{arg}}(h)} \wedge |h|\rightarrow 0\}$.
Then by inspection of (\ref{Bog-qa}) and (\ref{GSB-qa}) we find that $\rm{(GSB)}_{q-a}$ and $\rm{(BEC)}_{q-a}$
are \textit{equivalent}:
\begin{eqnarray}\label{Bog=GSB-qa}
&&\lim_{h \rightarrow 0}\lim_{\Lambda} \ \omega_{\beta,{\mu}_{\Lambda}
(\beta,{\rho}, h),\Lambda,h}^{0}({b^{*}_{q}}/\sqrt{V}) \ \omega_{\beta,{\mu}_{\Lambda}
(\beta,{\rho}, h),\Lambda,h}^{0}({b_{q}}/\sqrt{V}) = \\
&& = \lim_{h \rightarrow 0}\lim_{\Lambda} \ \omega_{\beta,{\mu}_{\Lambda}
(\beta,{\rho}, h),\Lambda,h}^{0}({b^{*}_{q}b_{q}}/{V}) =  {\rho} - \rho_{c}(\beta) \ . \nonumber
\end{eqnarray}
Note that by (\ref{PBG-ODLRO}) the $\rm{(GSB)}_{q-a}$ and $\rm{(BEC)}_{q-a}$  are in turn \textit{equivalent} to
$\rm{(ODLRO)}_{q-a}$. In \textit{contrast} to $\rm{(BEC)}_{q-a}$ for the one-mode BEC one gets
\begin{equation*}
\lim_{\Lambda} \ \omega_{\beta,{\mu}_{\Lambda}
(\beta,{\rho}, 0),\Lambda,0}^{0}({b^{*}_{q}b_{q}}/{V}) =
\lim_{\Lambda} \ \omega_{\beta,{\mu}_{\Lambda}
(\beta,{\rho}, 0),\Lambda, 0}^{0}({b^{*}_{q}}/\sqrt{V}) \ \omega_{\beta,{\mu}_{\Lambda}
(\beta,{\rho}, 0),\Lambda, 0}^{0}({b_{q}}/\sqrt{V})= 0 \ ,
\end{equation*}
for any $\rho$ and $q\in \Lambda^{*}$ as soon as $\alpha_1 > 1/2$, see Section \ref{sec:gBEC}.
On the other hand, the value of gBEC coincides with $\rm{(BEC)}_{q-a}$.
\begin{rem}\label{rem:3.311}
Therefore, the zero-mode conventional BEC and the zero-mode quasi-average $\rm{(BEC)}_{q-a}$  for the perfect
Bose-gas are \textit{not} equivalent: $\rm{(BEC)}_{q-a}$ $\nRightarrow$  BEC ,
but the zero-mode $\rm{(BEC)}_{q-a}$ is \textit{equivalent} to gBEC: $\rm{(BEC)}_{q-a}$ $\Leftrightarrow$ gBEC.
The equivalence (\ref{PBG-qa-equiv}) shows that the Bogoliubov quasi-average method is definitely appropriate
for the case of the PBG.
\end{rem}
We comment that (\ref{4.15.1}), (\ref{PBG-LimSt-phi}) show that the states $\omega_{\beta,\mu,\phi}$ are not
gauge invariant. Assuming that they are the ergodic states in the ergodic decomposition of $\omega_{\beta,\mu}$,
it follows that for \textit{interacting} Bose-gas one has: $\rm{(BEC)}_{q-a}$ $\Leftrightarrow$ $\rm{(GSB)}_{q-a}$,
which is similar to the equivalence for the PBG. It is illuminating
to observe the explicit mechanism for the appearance of the breaking symmetry phase $\phi$, connected with
(\ref{4.12.1}) of Proposition \ref{prop:4.1} in the PBG case. Note that in this case the chemical potential
remains proportional to $|\lambda|$ even after the thermodynamic limit (\ref{4.14.3}). This property
persists also for the {interacting} Bose-gas, see Section \ref{sec:BogAppr-Q-A}.
\subsection{Interaction, quasi-averages and the Bogoliubov $c$-number approximation} \label{sec:BogAppr-Q-A}
We now consider the imperfect Bose-gas with interaction(\ref{4.1})-(\ref{4.3}). The famous
\textit{Bogoliubov approximation} that replacing $\eta_{\Lambda, 0}(b), \eta_{\Lambda, 0}(b^{*})$
(\ref{b-p-mode-av}) by $c$-numbers \cite{Bog07} (see also \cite{ZBru}, \cite{JaZ10}, \cite{Za14}) will be
instrumental. The exactness of this procedure was proved by Ginibre \cite{Gin} on the level of thermodynamics.
Later Lieb, Seiringer and Yngvason (\cite{LSYng1}, \cite{LSYng}) and independently
S\"{u}t\"{o} \cite{Suto1} improved the arguments in \cite{Gin} and elucidated the \textit{exactness}
of the Bogoliubov approximation. In our analysis we shall rely on the method of \cite{LSYng}, which
uses the Berezin-Lieb inequality \cite{Lieb1}.

Recall that the Fock space ${\cal F}_{\Lambda} \simeq {\cal F}_{{0}} \otimes {\cal F}^{\prime}$, where
${\cal F}_{{0}}$ denotes the zero-mode subspace and ${\cal F}^{\prime} := {\cal F}_{{k} \ne {0}}$, see
Section \ref{sec:gBEC}.
Let $z\in \mathbb{C}$ be a complex number and  $|z\rangle = \exp(-|z|^{2}/2 +z b_{{0}}^{*})\ |0\rangle$ be
the Glauber coherent vector in ${\cal F}_{{0}}$. As in \cite{LSYng}, let operator
$(H_{\Lambda,\mu,\lambda})^{'}(z)$ be the \textit{lower symbol} of the operator
$H_{\Lambda,\mu,\lambda}$ (\ref{4.6}).
Then the corresponding to this symbol pressure $p_{\beta,\Lambda,\mu,\lambda}^{'}$ is defined by
\begin{equation}
\exp(\beta V p_{\beta,\Lambda,\mu,\lambda}^{'})=\Xi_{\Lambda}(\beta,\mu,\lambda)^{'} =
\int_{\mathbb{C}} d^{2}z {\rm{Tr}}_{{\cal F}^{'}} \exp(-\beta (H_{\Lambda,\mu,\lambda})^{'}(z)) \ .
\label{2.4.18}
\end{equation}

Consider the probability density:
\begin{equation}
{\cal W}_{\mu,\Lambda, \lambda}(z) := \Xi_{\Lambda}(\beta,\mu,\lambda)^{-1}\\
{\rm{Tr}}_{{\cal F}^{'}}\langle z| \exp(-\beta H_{\Lambda,\mu,\lambda})|z \rangle \ .
\label{2.4.19}
\end{equation}
As it is proved in \cite{LSYng} for almost all $\lambda >0$ the density
${\cal W}_{\mu,\Lambda, \lambda} (\zeta \sqrt{V})$ converges, as $V \to \infty$, to $\delta$-density at
the point $\zeta_{max}(\lambda)=\lim_{V \to \infty} {z_{max}(\lambda)}/{\sqrt{V}}$, where $ z_{max}(\lambda)$
maximises the partition function ${\rm{Tr}}_{{\cal F}^{'}} \exp(-\beta (H_{\Lambda,\mu,\lambda})^{'}(z))$.
Although \cite{LSYng} took $\phi=0$ in (\ref{4.8}), their results in the general case (\ref{4.8}) may be obtained
by the trivial substitution $b_{{0}}\to b_{{0}}\exp(-i\phi)$, $b_{{0}}^{*} \to b_{{0}}^{*} \exp(i\phi)$ motivated
by (\ref{4.6}). Note that expression (34) in \cite{LSYng} may be thus re-written as
\begin{eqnarray}
&& \lim_{V \to \infty} \omega_{\beta,\mu,\Lambda,\lambda}(\eta_{\Lambda, 0}(b^{*}\exp(i\phi))=
\lim_{V \to \infty} \omega_{\beta,\mu,\Lambda,\lambda}(\eta_{\Lambda,0}(b\exp(-i\phi))  \nonumber \\
&& = \zeta_{max}(\lambda)=\frac{\partial p(\beta, \mu,\lambda)}{\partial \lambda} \ , \label{2.4.20}
\end{eqnarray}
and consequently yields
\begin{equation}
\lim_{V \to \infty} \omega_{\beta,\mu,\Lambda,\lambda}(\eta_{\Lambda, 0}(b^{*})\eta_{\Lambda, 0}(b))
= |\zeta_{max}(\lambda)|^{2} \ .
\label{2.4.21}
\end{equation}
Here we denote by
\begin{equation}\label{4.22}
p(\beta,\mu,\lambda) = \lim_{V \to \infty} p_{\beta,\mu,\Lambda,\lambda} \ ,
\end{equation}
the grand-canonical pressure of the imperfect Bose-gas (\ref{4.1})-(\ref{4.3}) in the thermodynamic limit.
Equality (\ref{2.4.20}) follows from the convexity of $p_{\beta,\mu,\Lambda,\lambda}$ in
$\lambda = |\lambda_{\phi}|$ by the Griffiths lemma \cite{Gri66}.
In \cite{LSYng} it is shown the pressure $p(\beta,\mu,\lambda)$ is equal to
\begin{equation}\label{4.23}
p(\beta,\mu,\lambda)^{'} = \lim_{V \to \infty} p_{\beta,\mu,\Lambda,\lambda}^{'} \ .
\end{equation}
Moreover, (\ref{4.22}) is also equal to the pressure $p(\beta,\mu,\lambda)^{''}$, which is the thermodynamic limit
of the pressure associated to the \textit{upper symbol} of the operator $H_{\Lambda,\mu,\lambda}$.

The crucial is the proof \cite{LSYng} that all of these \textit{three} pressures $p', p, p''$ coincide with
$p_{max}(\beta,\mu,\lambda)$, which is the pressure associated with
${\rm{max}}_{z} {\rm{Tr}}_{{\cal F}^{'}} \exp(-\beta (H_{\Lambda,\mu,\lambda})^{'}(z))$:
\begin{equation}\label{4.231}
 p_{max}(\beta,\mu,\lambda)= \lim_{V \to \infty}
 \frac{1}{\beta V}\ln \{{\rm{max}}_{z} {\rm{Tr}}_{{\cal F}^{'}} \exp(-\beta (H_{\Lambda,\mu,\lambda})^{'}(z))\} \ .
\end{equation}

Now we are in position to prove one of the main statements of this paper.
\begin{theo}\label{theo:3.4}
Consider the system of interacting Bosons (\ref{4.1})-(\ref{4.8}). If this system displays
$\rm{(ODLRO)}_{q-a}$/$\rm{(BEC)}_{q-a}$, then the limit $\omega_{\beta,\mu,\phi}:=
\lim_{\lambda \to +0} \lim_{V \to \infty}\omega_{\beta,\mu,\Lambda,\lambda_{\phi}} $, on the
set of monomials $\{\eta_{0}(b^{*})^{m}\eta_{0}(b)^{n}\}_{m,n \in \mathbb{N}\cup 0}$ exists and satisfies
\begin{equation}\label{4.24.1}
\omega_{\beta,\mu,\phi} (\eta_{0}(b^{*})) = \sqrt{\rho_{0}} \exp(i\phi) \ ,
\end{equation}
\begin{equation}\label{4.24.2}
\omega_{\beta,\mu,\phi} (\eta_{0}(b)) = \sqrt{\rho_{0}} \exp(-i\phi) \ ,
\end{equation}
together with $\rm{(GSB)}_{q-a}$:
\begin{equation}\label{3.4.24.3}
\omega_{\beta,\mu,\phi} (\eta_{0}(b^{*})\eta_{0}(b)) =
\omega_{\beta,\mu,\phi} (\eta_{0}(b^{*})) \  \omega_{\beta,\mu,\phi}(\eta_{0}(b))
= \rho_{{0}} \ , \ \forall \phi \in [0,2\pi) \ ,
\end{equation}
and
\begin{equation}\label{4.24.4}
\omega_{\beta,\mu} = \frac{1}{2\pi} \int_{0}^{2\pi} d\phi \ \omega_{\beta,\mu,\phi} \ .
\end{equation}
On the Weyl algebra the limit that defines $\omega_{\beta,\mu,\phi}, \ \phi \in [0,2\pi)$
exists along the nets in variables $(\lambda,V)$. The corresponding states are ergodic,
and coincide with the states obtained in Proposition \ref{prop:A.1}.

Conversely, if the $\rm{(GSB)}_{q-a}$ occurs in the sense that (\ref{4.24.1}), (\ref{4.24.2}) hold with
$\rho_{0} \ne 0$, then one gets that $\rm{(ODLRO)}_{q-a}$/$\rm{(BEC)}_{q-a}$ take place.
\end{theo}
\begin{proof} We only have to prove the direct statement, because the converse follows by applying the Schwarz
inequality to the states $\omega_{\beta,\mu,\phi}$, together with the forthcoming (\ref{3.4.27}).

We thus prove $\rm{(ODLRO)}_{q-a}$ $\Rightarrow$ $\rm{(GSB)}_{q-a}$. We first assume that some state
$ \omega_{\beta,\mu,\phi_{0}},\phi_{0}
\in [0,2\pi)$ satisfies $\rm{(ODLRO)}_{q-a}$. Then by (\ref{2.4.21}),
\begin{equation}\label{2.4.25.1}
\lim_{\lambda \to +0} \lim_{V \to \infty}
\omega_{\beta,\mu,\Lambda,\lambda}(\eta_{\Lambda,0}(b^{*})\eta_{\Lambda,0}(b))
= \lim_{\lambda \to +0} |\zeta_{max}(\lambda)|^{2} =: \rho_{{0}} > 0 \ .
\end{equation}
The above limit exists by the convexity of $p(\beta,\mu,\lambda)$ in $\lambda$ and (\ref{4.14.1}) by virtue of
(\ref{2.4.25.1}),
\begin{equation}\label{4.25.2}
\lim_{\lambda \to +0} \frac{\partial p(\beta,\mu,\lambda)}{\partial \lambda} \ne 0 \ .
\end{equation}
At the same time, (\ref{2.4.20}) shows that all states $\omega_{\beta,\mu,\phi}$ satisfy (\ref{2.4.25.1}).
Thus, $\rm{(GSB)}_{q-a}$ is broken in the states $\omega_{\beta,\mu,\phi}, \phi \in [0,2\pi)$.\,
We now prove that the original assumption (\ref{4.16}) implies that all states
$\omega_{\beta,\mu,\phi}, \phi \in [0,2\pi)$ exhibit $\rm{(ODLRO)}_{q-a}$.

Gauge invariance of $\omega_{\beta,\mu,\Lambda}$ (or equivalently $H_{\Lambda,\mu}$) yields, by (\ref{4.7}),
(\ref{4.17}),
\begin{equation}
\omega_{\beta,\mu,\Lambda,\lambda}(\eta_{\Lambda,0}(b^{*})\eta_{\Lambda,0}(b))
=\omega_{\beta,\mu,\Lambda,-\lambda}(\eta_{\Lambda,0}(b^{*})\eta_{\Lambda,0}(b)) \ .
\label{2.4.26.1}
\end{equation}
Again by (\ref{4.7}), (\ref{4.12.1}) and gauge invariance of $H_{\Lambda,\mu}$,
\begin{equation*}
\lim_{\lambda \to -0} \frac{\partial p(\beta,\mu,\lambda)}{\partial \lambda}=
-\lim_{\lambda \to +0} \frac{\partial p(\beta,\mu,\lambda)}{\partial \lambda} \ ,
\end{equation*}
and, since by convexity the derivative ${\partial p(\beta,\mu,\lambda)}/{\partial \lambda}$ is monotone
increasing, we find
\begin{equation}
\lim_{\lambda \to +0} \frac{\partial p(\beta,\mu,\lambda)}{\partial \lambda}
= \lim_{\lambda \to +0} \zeta_{max}(\lambda) = \sqrt{\rho_{0}} \ ,
\label{2.4.26.2}
\end{equation}
\begin{equation}
\lim_{\lambda \to -0} \frac{\partial p(\beta,\mu,\lambda)}{\partial \lambda}
= -\lim_{\lambda \to +0} \zeta_{max}(\lambda)= -\sqrt{\rho_{0}} \ .
\label{2.4.26.3}
\end{equation}
Again by (\ref{2.4.26.1}),
\begin{equation}
\lim_{\lambda \to -0}\lim_{V \to \infty}
\omega_{\beta,\mu,\Lambda,\lambda}(\eta_{\Lambda,0}(b^{*})\eta_{\Lambda,0}(b))
= \lim_{\lambda \to +0}\lim_{V \to \infty}
\omega_{\beta,\mu,\Lambda,\lambda}(\eta_{\Lambda,0}(b^{*})\eta_{\Lambda,0}(b)) \ .
\label{2.4.26.4}
\end{equation}
By \cite{LSYng}, the weight ${\cal W}_{\mu,\lambda}$ is, for $\lambda=0$, supported on a disc with radius
equal to the right-derivative (\ref{4.25.2}). Convexity of the pressure as a function of $\lambda$ implies
\begin{eqnarray*}
\frac{\partial p(\beta,\mu,\lambda_{0}^{-})}{\partial \lambda_{0}^{-}} \le \lim_{\lambda \to -0}\frac{\partial
p(\beta,\mu,\lambda)}{\partial \lambda}
\le \lim_{\lambda \to +0}\frac{\partial p(\beta,\mu,\lambda)}{\partial \lambda} \le \frac{\partial p(\beta,\mu,
\lambda_{0}^{+})}{\partial \lambda_{0}^{+}} \ ,
\end{eqnarray*}
for any $\lambda_{0}^{-}<0<\lambda_{0}^{+}$. Therefore, by the Griffiths lemma (see e.g. \cite{Gri66},
\cite{LSYng}) one gets
\begin{eqnarray}
\lim_{\lambda \to -0}\lim_{V \to \infty}
\omega_{\beta,\mu,\Lambda,\lambda}(\eta_{\Lambda,0}(b^{*})\eta_{\Lambda,0}(b))
&\le& \lim_{V \to \infty} \omega_{\beta,\mu,\Lambda}(\frac{b_{{0}}^{*}b_{{0}}}{V}) \nonumber \\
&\le& \lim_{\lambda \to +0}\lim_{V \to \infty}
\omega_{\beta,\mu,\Lambda,\lambda}(\eta_{\Lambda,0}(b^{*})\eta_{\Lambda,0}(b)) \ .
\label{3.4.27}
\end{eqnarray}
Then (\ref{2.4.26.4}) and (\ref{3.4.27}) yield
\begin{equation}
\lim_{V \to \infty} \omega_{\beta,\mu,\Lambda}(\frac{b_{{0}}^{*}b_{{0}}}{V})=\\
\lim_{\lambda \to +0}\lim_{V \to \infty}
\omega_{\beta,\mu,\Lambda,\lambda}(\eta_{\Lambda,0}(b^{*})\eta_{\Lambda,0}(b)) \ , \  \
\forall \phi \in [0,2\pi) \ .
\label{3.4.28}
\end{equation}
This proves that all $\omega_{\beta,\mu,\phi}, \phi \in [0,2\pi)$ satisfy $\rm{(ODLRO)}_{q-a}$, as asserted.

By (\ref{2.4.20}) and (\ref{2.4.26.2}) one gets (\ref{4.24.1}) and (\ref{4.24.2}). Then (\ref{4.24.4}) is a
consequence of the gauge-invariance of $\omega_{\beta,\mu}$. Ergodicity of the states
$\omega_{\beta,\mu,\phi}, \phi \in [0,2\pi)$ follows from
(\ref{4.24.1}), (\ref{4.24.2}), and (\ref{3.4.28}), see  Definition \ref{SSB}(ii).

An equivalent construction is possible using the \textit{Weyl algebra} instead of the \textit{polynomial algebra},
see \cite{Ver}, Ch.4.3.2, and references given there for Proposition \ref{prop:A.1}.
The limit along a subnet in the $(\lambda,V)$ variables exists by weak*-compactness, and,
by asymptotically abelianness of the Weyl algebra  for space translations (see, e.g., \cite{BR97}, Example 5.2.19),
the ergodic decomposition (\ref{4.24.4}), which is also a central decomposition, is unique. Thus, the
$\omega_{\beta,\mu,\phi}, \phi \in [0,2\pi)$ coincide with the states constructed in Proposition \ref{prop:A.1}.
\end{proof}
\begin{rem}\label{rem:3.5}
Our Remark \ref{rem:3.3} and Theorem \ref{theo:3.4} elucidate a problem discussed in \cite{LSYng}.
In this paper the authors defined a generalised  Gauge Symmetry Breaking via quasi-average $\rm{(GSB)}_{q-a} \ $,
i.e. by
$\lim_{\lambda \to +0} \lim_{V \to \infty} \omega_{\beta,\mu,\Lambda,\lambda}(\eta_{\Lambda, 0}(b)) \ne 0 \ $.
(If it involves other than the gauge group, we denote this by $\rm{(SSB)}_{q-a}$.)
Similarly they modified definition of the one-mode condensation denoted by $\rm{(BEC)}_{q-a}$ (\ref{2.4.25.1}),
and established the equivalence: $\rm{(GSB)}_{q-a} \Leftrightarrow \rm{(BEC)}_{q-a}$.
They also posed a problem: whether $\rm{(BEC)}_{q-a} \Leftrightarrow \rm{BEC}\ $?

In Theorem \ref{theo:3.4} we show that $\rm{(GSB)}_{q-a}$ (\ref{4.24.1}) implies $\rm{(ODLRO)}_{q-a}$, or
$\rm{(BEC)}_{q-a}$.
Note that for the \textit{zero-mode} BEC (Definition \ref{defi:2.1}), the same theorem shows that their
question is answered in the \textit{affirmative}. This is due to the crucial fact that the state
$\omega_{\beta,\mu}$ is gauge-invariant, which is consistent with the decomposition (\ref{4.24.4}) and leads to
the inequalities (\ref{3.4.27}).

On the other hand, for another (but nonetheless equally important, as argued in (\ref{Int-TypeIII}))
types of condensation the comparison (implication "$\Rightarrow$", or equivalence "$\Leftrightarrow$", see
Remark \ref{rem:3.3}) between q-a and \textit{non} q-a values may \textit{fail}.

For example, we note that for PBG the value of $\rm{(BEC)}_{q-a}$ is strictly \textit{larger} the zero-mode
BEC for anisotropy $\alpha_1 \geq 1/2$, and $\rm{(BEC)}_{q-a} \nRightarrow $ BEC for $\alpha_1 > 1/2 $,
see Section \ref{sec:gBEC}. {{One observes the similar phenomenon
$\rm{(BEC)}_{q-a} \nRightarrow $ BEC for the interacting Bose-gas (\ref{Int-TypeIII}). Although for the
both cases (PBG and (\ref{Int-TypeIII})) we get $\rm{(BEC)}_{q-a} \Leftrightarrow $ gBEC, see
Section \ref{sec:gBEC-BQ-A}.}} Therefore, in the general case the answer to the question in Remark \ref{rem:3.5}
is \textit{negative}.
Note that the fact established in Theorem \ref{theo:3.4}, that the quasi-averages lead to ergodic states
clarifies an important conceptual aspect of the quasi-average \textit{trick}.
\end{rem}

\begin{rem}\label{rem:4.21}
The states $\omega_{\beta,\mu,\phi}$ in Theorem \ref{theo:3.4} have the property ii) of Proposition \ref{prop:A.1},
i.e., if $\phi_{1} \ne \phi_{2}$, then
$\omega_{\beta,\mu,\phi_{1}} \ne \omega_{\beta,\mu,\phi_{2}}$. By a theorem of Kadison \cite{Kadison}, two factor
states are either disjoint or quasi-equivalent,
and thus the states $\omega_{\beta,\mu,\phi}$ for different $\phi$ are mutually disjoint.
This phenomenon also occurs for spontaneous magnetisation in quantum spin systems.
It is in this sense that the word ''degeneracy'' must be understood, compare with the discussion in \cite{Bog70}.
\end{rem}
\section{Bogoliubov quasi-averages and critical quantum fluctuations} \label{BQ-A-QFl}
The aim of this section is to show that the scaled breaking symmetry external sources may have a nontrivial
impact on critical quantum fluctuations. This demonstrates that quasi-averages are helpful not only to study
phase transitions via $\rm{(SSB)}_{q-a}$, but also to analyse the corresponding \textit{critical} and, in
particular, commutative and noncommutative \textit{quantum} fluctuations. To this end we use for illustration
an example of a concrete model that manifests quantum phase transition with discrete $\rm{(SSB)}_{q-a}$
\cite{VZ1}
\subsection{Algebra of fluctuation operators.}\label{Alg-Fl-Op}

We start this section by a general setup to recall the concept of \textit{quantum fluctuations} via the
noncommutative Central Limit Theorem (CLT) and the corresponding to them Canonical Commutation Relations (CCR).

To describe any ($\mathbb{Z}^d$-lattice) quantum statistical model, one has to
start from {\it microscopic} dynamical system, which is a triplet ($\mathcal{A},\omega,\alpha_t$) where:\\
\hspace*{1 cm}(a) $\mathcal{A}=\cup_\Lambda \mathcal{A}_\Lambda$ is the quasi-local algebra of observables,
here $\Lambda$ are bounded subset of $\mathbb{Z}^d$ and $[\mathcal{A}_{\Lambda'},\mathcal{A}_{\Lambda''}]=0$ if
$\Lambda'\cap\Lambda''=\emptyset $.\\
\hspace*{1 cm}(b) $\omega$ is a {\it state} on $\mathcal{A}$.  Let $\tau_x$ be space
translation automorphism of translations over the distance $x\in\mathbb{Z}^d$,
i.e., the map $\tau_x:A\in \mathcal{A}_\Lambda\rightarrow\tau_x(A)\in \mathcal{A}_{\Lambda+x}$. Then the
state $\omega$ is {\it translation-invariant} if $\omega\circ\tau_x(A)\equiv\omega(\tau_x(A))=\omega(A)$ and
{\it space-clustering} if $\lim_{\vert x\vert\rightarrow\infty}\omega(A\tau_x(B))= \omega(A)\omega(B)$ for
$A,B\in\mathcal{A}$.\\
\hspace*{1 cm}(c) $\alpha_t$ is dynamics described by the family of local Hamiltonians
$\{{H}_\Lambda\}_{\Lambda \subset \mathbb{Z}^d}$. Usually, $\alpha_t$ is defined as a norm limit of
the local dynamics:
$\alpha_t(A):= \lim_\Lambda\exp(it{H}_\Lambda)A\exp(-it{H}_\Lambda)$, i.e.,
$\alpha_t:\mathcal{A}\rightarrow\overline{\cal A}$-norm-closure
of $\mathcal{A}$.  For {\it equilibrium states} one assumes that $\omega\circ\alpha_t=\omega$
({\it time invariance}).

Note that usually one assumes also that the space and time translations {\it commute}:
$\tau_x(\alpha_t(A))=\alpha_t(\tau_x((A))$, where $A\in \mathcal{A}_{\Lambda}$
and $\Lambda \subset \mathbb{Z}^d$.

On the way from the {\it micro system} $(\mathcal{A}, \, \omega,\alpha_t)$ to {\it macro
system} of physical observables, one has to distinguish \textit{two} essentially different classes.

The first one ({\textit{macro} I}) corresponds to the \textit{Weak Law of Large Numbers} (WLLN).
It is well-suited for description of {\it order parameters} in the system.
Formally this class of observables is defined as follows:
for any $A\in\mathcal{A}$ the local space \textit{mean} mapping $m_\Lambda:A\rightarrow
m_\Lambda(A):= \vert\Lambda\vert^{-1} \sum_{x\in\Lambda}\tau_x(A)$.
Then, the limiting map $m:A\rightarrow {\cal C}$
\begin{equation}\label{1.1}
m(A)= w\!-\!\lim_\Lambda m_\Lambda(A) \ , \ \forall A\in\mathcal{A} \ ,
\end{equation}
exists in the $\omega$-\textit{weak} topology, induced by the \textit{ergodic}, see (b), state $\omega$.

Let $m(\mathcal{A})=\{m(A):  A\in\mathcal{A}\}$.  Then the {\it macro system} I has the
following properties:\\
\hspace*{1 cm}(Ia) $m(\mathcal{A})$ is a set of observables {\it at infinity}  because
[$m(\mathcal{A}), \mathcal{A}]=0$.\\
\hspace*{1cm}(Ib) $m(\mathcal{A})$ is an {\it abelian} algebra and $m(A)=\omega(A) \cdot \I $. Hence
the states on $m(\mathcal{A})$ are probability measures.\\
\hspace*{1 cm}(Ic) Since $m(\tau_a(A)) = m(A)$, the map $m$: $\mathcal{A}\rightarrow m(\mathcal{A})$ is
\textit{not} injective. This is a mathematical expression of the {\it coarse graining} under the WLLN.\\
\hspace*{1 cm}(Id) The {\it macro-dynamics} $\tilde{\alpha}_t (m(A)):= m(\alpha_t (A))$ induced by the
micro-dynamics (c) on $m(\mathcal{A})$ is {\it trivial} since
$m(\alpha_t (A))=\omega(\alpha_t(A))\cdot \I=\omega(A)\cdot \I =m(A)$.

The second class of {\it macro-observable} (\textit{macro} II) correspond to the {\it
Quantum Central Limit} (QCL), which is well-suited for description of
(quantum) fluctuations and, in particular, for description of collective and
elementary excitations (\textit{phonons}, \textit{plasmons}, \textit{excitons}, etc) in many body
quantum systems \cite{Ver}.

To proceed in construction of the \textit{macro} II one has to be more precise.
Let $A\in\mathcal{A}_{sa}:=\{B\in\mathcal{A}:  B=B^\ast\}$
be self-adjoint operators in a Hilbert
space $\mathfrak{H}$.  Then one can define the local
mapping $F^\d_{k,\Lambda}:  A\rightarrow F^\d_{k,\Lambda}(A)$, where
\begin{equation}\label{1.2}
F^{\d_A}_{k,\Lambda}(A): =\frac{1}{\vert\Lambda\vert^{\frac{1}{2}+\d_A}}\sum_{x\in\Lambda}(\tau_x(A)-
\omega(A))e^{ikx} \, , \ \ k \, , \ \d_A \in \mathbb{R} \ .
\end{equation}
This is nothing but the {\it local fluctuation operator} for the mode k.  If $\d_A =0$, this fluctuation
operator is called {\it normal}.  The next important concept is due to \cite{GVV1}-\cite{GV} and a further
development in \cite{Re}:\\
\textbf{Quantum Central Limit Theorem.} Let
\begin{equation*}
\g_\omega(r):=\sup_{\Lambda,\Lambda'}\sup_{A\in {\cal
A}_\Lambda\atop B\in {\cal A}_{\Lambda'}}\left\{\frac
{\omega(AB)-\omega(A)\omega(B)}{\parallel A\parallel \parallel
B\parallel}:\;r\leq dist( \Lambda,\Lambda')\right\} \ \ {\rm{and}} \ \
\sum_{x\in\mathbb{Z}^d}\g_\omega(\vert x\vert)<\infty \ .
\end{equation*}
Then, for any $A\in\mathcal{A}_{sa}$, the corresponding limiting \textit{characteristic} function exists
for the normal fluctuation operator ($\d_A=0$) for the \textit{zero-mode} $k=0$:
\begin{equation}\label{1.3}
\lim_\Lambda\omega(e^{i\, u \, F_\Lambda(A)})=e^{- {u^2}S_\omega(A,A)/{2}} \ \ ,\ u\in\mathbb{R} \ ,
\end{equation}
where sesquilinear form
$S_\omega(A,B):={\rm{Re}}\sum_{x\in\mathbb{Z}^d}\omega((A-\omega(A)) \, \tau_x(B-\omega(B)))$, for
$A,B \in\mathcal{A}_{sa}$.
\\
\hspace*{1 cm}(IIa) The result (\ref{1.3}) establishes the meaning of the QCL
for normal fluctuation operators. If (\ref{1.3}) exists for $\d_{A,B} \not= 0$ with the
modified sesquilinear form
\begin{equation}\label{1.4}
S_{\omega,\d_{A,B}}(A,B)=\lim_\Lambda {\rm{Re}}\frac{1}{\vert \Lambda\vert^{\d_A+\d_B}}
\sum_{x\in\mathbb{Z}^d}\omega((A-\omega(A))\ \tau_x(B-\omega(B))) \ ,
\end{equation}
we say that QCL exists for the zero-mode {\it abnormal fluctuations}:
\begin{equation}\label{1.5}
\lim_\Lambda F^{\d_A}_\Lambda(A)=F^{\d_A}(A) \ .
\end{equation}
The fluctuation operators $\{F^{\d_A}(A)\}_{A \in\mathcal{A}_{sa}}$ act in a Hilbert space $\mathcal{H}$,
which is defined by the corresponding to (\ref{1.3}) and (\ref{1.4}) Reconstruction Theorem.\\
\hspace*{1 cm}(IIb) To this end we consider $\mathcal{A}_{as}$ as a {\it vector-space} with {\it symplectic} form
$\sigma_\omega(\cdot,\cdot)$, which is correctly defined for the case $\d_A+\d_B=0$ by the WLLN :
\begin{equation}\label{1.6}
i\sigma_\omega(A,B)\cdot \I=\lim_\Lambda[F^{\d_A}_\Lambda(A),F^{\d_B}_\Lambda(B)]=
2\, i \, {\rm{Im}}\sum_{y\in\mathbb{Z}^d}(\omega(A\tau_y(B))-\omega(A)\omega(B)) \ .
\end{equation}
Suppose that $W(\mathcal{A}_{sa},\sigma_\omega)$ is the {\it Weyl algebra}, i.e., the family of the Weyl operators
$W\, : \, \mathcal{A}_{sa}\ni A \mapsto W(A)$ such that
\begin{equation}\label{1.7}
W(A)W(B)=W(A+B)e^{- i \, \sigma_\omega(A,B)/2} \ ,
\end{equation}
where operators $A,B\in\mathcal{A}_{sa}$, acting in the Hilbert space $\mathfrak{H}$.\\
\textbf{Reconstruction Theorem.}  Let $\tilde{\omega}$ be a {\it quasi-free} state on the Weyl algebra
$W(\mathcal{A}_{sa},\sigma_\omega)$, which is defined by the sesquilinear form $S_\omega(\cdot,\cdot)$:
\begin{equation}\label{1.8}
\tilde{\omega}(W(A)):=e^{- S_\omega(A,A)/2} \ .
\end{equation}
Since (\ref{1.7}) implies that $W(A):=e^{i\Phi (A)}$, where $\Phi : A \mapsto \Phi (A)$ are {\it boson
field} operators acting in the Canonical Commutation Relations (CCR) representation Hilbert space
${\cal H}_{\tilde{\omega}}$ corresponding to the state $\tilde{\omega}$,
the relations (\ref{1.2})-(\ref{1.8}) yield {\it identifications} of the spaces:
${\cal H}={\cal H}_{\tilde{\omega}}$, and of the operators:
\begin{equation}\label{1.9}
\lim_\Lambda F^{\d_A}_{\Lambda}(A) =: F^{\d_A}(A)=\Phi (A) \ .
\end{equation}
\hspace*{1cm}(IIc) The {\it Reconstruction Theorem} gives a transition from the
micro-system ($\mathcal{A}_{sa},\omega$) to the macro-system of {\it fluctuation
operators} ($F(\mathcal{A}_{sa},\sigma_\omega),\tilde{\omega}$). We note that
$F(\mathcal{A}_{sa},\sigma_\omega)=\{F^{\d_A}(A)\}_{A\in {\cal H}_{sa}}$ is the
{\it CCR-algebra} on the symplectic space $(\mathcal{A}_{sa},\sigma_\omega)$, see
(\ref{1.7})-(\ref{1.9}).\\
\hspace*{1 cm}(IId) The map $F:\mathcal{A}_{sa}\rightarrow
F(\mathcal{A}_{sa},\sigma_\omega)$ is not injective (the zero-mode coarse graining).
For example, $\tilde{\tau}_x(F(A)):= F(\tau_x(A))=F(A)$, but it has a non-trivial
{\it macro-dynamics} $\tilde{\alpha}_t(F(A)):= F(\alpha_t(A))$.
Therefore, the {\it macro-system} II defined by the algebra of fluctuation
operators is the {\it triplet}
($F(\mathcal{A}_{sa},\sigma_\omega),\tilde{\omega},\tilde{\alpha}_t)$.

Identification of the algebra of the fluctuation operators $F$(\h$_{sa},\sigma_\omega)$ for a given
micro-system (\h,$\omega,\alpha_t$) with the CCR-algebra of the boson field operators
supplies a mathematical description of so-called {\it collective
excitations} (phonons, plasmons, excitons etc) in the \textit{pure} state $\omega$.

The same approach gives as well a break into the mathematical foundation of
another physical concept:  the {\it Linear Response Theory} \cite{GVV2}.
In the latter case, it became clear that {\it algebra of fluctuations} is more
sensible with respect to "gentle" perturbations of the microscopic
Hamiltonian by \textit{external sources} than, e.g., {\it algebra at infinity} $m($\h).  This property
gets even more sound if the equilibrium state $\omega$ (being pure) belongs
to the critical domain \cite{VZ1}. In this case, perturbations of microscopic
Hamiltonien, which do not change equilibrium state $\omega$ ("gentle" pertubations), can produce
{\it different} algebras of fluctuations independent of quantum or classical nature of the micro-system.

As we learned in Section \ref{sec:gBEC-BQ-A} the idea of perturbation of Hamiltonien to produce
{\it pure} equilibrium states comes back to the Bogoliubov quasi-averages.
Later this method was generalised to include the construction of the {\it mixed states} \cite{BZT}.
We recall that it can be formulated as follows:

(i) Let $\{B_l=\tau_l \,(B)\}_{l\in\mathbb{Z}}$ be operators breaking the symmetry of the initial system
\begin{equation*}
H_\Lambda(h) := H_\Lambda-\sum_{l\in\Lambda}h_lB_l\;\;,\;\;\; h_l\in\reel^1 \ .
\end{equation*}

(ii) Then the limiting states for $h_l=h$
\begin{equation}\label{eq:1.10}
\l-\r=\lim_{h\rightarrow 0}\lim_{\Lambda}\l-\r_{\Lambda,h} \ ,
\end{equation}
pick out pure states with respect to decomposition corresponding the
symmetry (Bogoliubov's quasi-averages).

(iii) If the external field $h=\widehat{h}/\vert \Lambda\vert^\alpha$, then the obvious generalization
of (\ref{eq:1.10}) either coincides with \textit{pure} states $(\alpha<\alpha_c)$ or give a
family of \textit{mixed} states enumerated by $\widehat{h}$ and $\alpha\geq\alpha_c$, see \cite{BZT}.

As it was found in \cite{VZ1}, the algebra of fluctuations for a
quantum model of ferroelectric (\textit{structural} phase transitions) depends on the
parameter $\alpha$ in the critical domain (below the critical line) even for
the pures states, i.e., for $\a<\a_c=1$ one obtains for correlation critical exponents (\ref{1.2}):
$\d_Q=\a/2$, while $\d_P=0$ (for $T\not= 0$, $T$ is the temperature).  Here $A:=Q$ and
$B:=P$ are respectively the atomic displacement and momentum operators in the
site ($l=0$) of $\mathbb{Z}$. The second observation of \cite{VZ1}
concerns the {\it quantum nature} of the critical fluctuations $F^{\d}(\cdot)$, i.e.
fluctuations in the \textit{pure} state $\omega$, which belongs to the critical line.
It was shown that expected abelian properties of critical fluctuations
can changes into non-abelian commutations between $F^{\d_Q}(Q)$ and
$F^{\d_P}(P)$ with $\d_Q=-\d_P>0$, at the quantum critical point $(T=0, \lambda =\lambda_c)$.  Here,
$\lambda :=\hbar/\sqrt{m}$ is the \textit{quantum parameter} of the model, where $m$ is the
mass of atoms in the nodes of lattice $\mathbb{Z}$.

Since usually one has a long-range correlations on the critical line, the {\it critical} fluctuations are
anticipated to be sensitive with respect to the above "gentle" perturbations
$\ h=\widehat{h}/\vert \Lambda\vert^\alpha$. On the other hand, they have to be also sensitive to
decay of a \textit{direct} interaction between particles:  in our model, the decay of
the harmonic force matrix elements is given by
\begin{equation}\label{eq:1.11}
\phi_{l,l'}\sim\vert l-l'\vert^{-(d+\s)}\;\; {\rm{for}} \;\;\vert l-l'\vert \longrightarrow \infty \ .
\end{equation}
If $\sigma\geq 2$, then one classifies interaction (\ref{eq:1.11}) as a {\it short-range}, whereas the
case $0<\sigma<2$ as {\it long-range}, because the corresponding lattice Fourier-transform has the
following two types of asymptotics for $k\rightarrow 0$:
\begin{equation}\label{eq:1.12}
\tilde{\phi}(k)\sim\left\{\begin{array}{ll} a^\sigma k^\sigma+o(k^\sigma) \ ,
& 0<\sigma<2 \ , \\ a^2k^2+o(k^2) \ , & \sigma\geq 2 \ . \end{array} \right.
\end{equation}

Therefore, our purpose is to find exponents $\d_A$ as the function of the
parameter $\a$ and $\sigma$ for a quantum ferroelectric model.  Note that $\d_Q=\d_Q(\a,\sigma)$ is
directly related to the critical exponent $\eta$ describing decay of the two-point correlation function
for displacements on the critical line: $\eta=2-2d\d_A$ ,  \cite{APS}.


\subsection{Quantum phase transition, fluctuations and quasi-averages} \label{QPT-Fluct-Q-A}
Let $\mathbb{Z}$ the d-dimensional square lattice.  At each lattice site $l$
occupied by a particle with mass $m$, we associate the position operator
$Q_l\in\reel^1$ and the momentum operator $P_l=(\hbar/i)(\partial/\partial
Q_l)$ in the Hilbert space $\mathcal{H}_l = L^2(\reel^1, dx)$.  Let $\Lambda$ be a finite cubic
subset of $\mathbb{Z}$, $V=\vert\Lambda\vert$ and the set $\Lambda^\ast$ is \textit{dual} to $\Lambda$
with respect to \textit{periodic} boundary conditions.  The local Hamiltonian $H_\Lambda$ of the model is
a self-adjoint operator on domain ${\rm{dom}}(H_\Lambda) \subset \mathcal{H}_\Lambda$, given by
\begin{equation}\label{2.1}
H_\Lambda=\sum_{l\in\Lambda}\frac{P_l^2}{2m}+\frac{1}{4}\sum_{l, \, l' \, \in \Lambda}\phi_{l,l'}(Q_l-
Q_{l'})^2+\sum_{l\in \Lambda}U(Q_l)-h\sum_{l\in \Lambda}Q_l \ .
\end{equation}
Here the local Hilbert space $\mathcal{H}_\Lambda := \otimes_{l \in \Lambda} \mathcal{H}_l $.
Note that the second term of (\ref{2.1}) represents the harmonic interaction between particles, the last term
represents the action of an external field and the third one is the
anharmonic on-site potential acting in each $l \in \mathbb{Z}$. Recall that potential $U$ must have a
\textit{double-well} form to describe a \textit{displacive} structural phase transition
attributed to the \textit{one}-component ferroelectric \cite{APS}. For example:
$U(x)=\frac{a}{2}Q_l^2+W(Q_l^2)$, $a<0$, with $W(x)=\frac{1}{2}bx^2$, $b>0$. Another example is
is a nonpolynomial $U$, such that $a>0$ and $W(x)=\frac{1}{2}b\exp{(-\eta x)}\;,\;\eta>0$ for $b>0$.
Then (\ref{2.1}) becomes
\begin{equation}\label{2.2}
H_\Lambda=\sum_{l\in \Lambda}\frac{P_l^2}{2m}+\frac{1}{4}\sum_{l, \, l' \, \in \Lambda}\phi_{l,l'}
(Q_l-Q_{l'})^2+\frac{a}{2}\sum_{l\in \Lambda}Q_l^2 \ + $$ $$ + \ \sum_{l\in
\Lambda}W(Q_l^2)-h\sum_{l\in \Lambda}Q_l \ .
\end{equation}
Recall that model (\ref{2.2}) manifests a structural phase transition, breaking $Z_2$-symmetry
$\{Q_l \rightarrow - Q_l\}_{l\in \mathbb{Z}}$ at low temperature, if the quantum parameter
$\lambda < \lambda_c$,  \cite{MPZ}, \cite{AKKR}.

We comment that a modified model (\ref{2.2}) can be solved exactly if one applies the following
approximation:
\begin{equation*}
\sum_{l\in\Lambda} W(Q_l^2)\longrightarrow  V \, W(\frac{1}{V}\sum_{l\in\Lambda}Q_l^2) \ ,
\end{equation*}
known as the concept of \textit{self-consistent} phonons (SCP), see \cite{APS}. This yields a model with
Hamiltonian
\begin{equation}\label{2.3}
H_\Lambda^{SCP}=\sum_{l\in
\Lambda}\frac{P_l^2}{2m}+\frac{1}{4}\sum_{l,\, l' \, \in \Lambda}\phi_{l,l'}
(Q_l-Q_{l'})^2+\frac{a}{2}\sum_{l\in \Lambda}Q_l^2 \ +$$
$$ \ + V \, W(\frac{1}{V}\sum_{l\in \Lambda}Q_l^2) - h \sum_{l\in \Lambda}Q_l \ ,
\end{equation}
that can be solved by the Approximating Hamiltonian Method (AHM) \cite{BBZKT}, see \cite{PT} and \cite{VZ1}.
Then the \textit{free-energy} density for Hamiltonian $H_\Lambda(c)$, which is
approximating for $H_\Lambda^{SCP}$ (\ref{2.3}), is
\begin{equation}\label{2.31}
f_{\Lambda}[H_\Lambda(c)] : =
 - \frac{1}{\beta V} \ln {\rm{Tr}}_{{\mathcal{H}}_{\Lambda}} e^{-\beta H_{\Lambda}(c)} \ , \ \ \
 \beta := \frac{1}{k_B T} \ ,
\end{equation}
where $k_B$ is the Boltzmann constant. Since the AHM yields that
\begin{equation}\label{2.32}
H_\Lambda(c) := \sum_{l\in
\Lambda}\frac{P_l^2}{2m}+\frac{1}{4}\sum_{l,\, l' \, \in \Lambda}\phi_{l,l'}
(Q_l-Q_{l'})^2+\frac{a}{2}\sum_{l\in \Lambda}Q_l^2 \ +$$
$$ \ + V \, \left[W(c) + W'(c)\left(\frac{1}{V}\sum_{l\in \Lambda}Q_l^2 - c\right)\right] -
h \sum_{l\in \Lambda}Q_l \ ,
\end{equation}
the \textit{free-energy} density (\ref{2.31}) gets the explicit form
$$ f_{\Lambda}[\H]=\frac{1}{\beta V}\sum_{q\in\Lambda^*}
\ln{\left[2\sinh{\frac{\beta\lambda\Omega_q(\C)}{2}}\right]}-
\frac{1}{2}\frac{h^2}{\Dc}$$
$$+[W(\C)-\C W'(\C)] \ . $$
Here $c =\C$ is a solution of the self-consistency equation:
\begin{equation}\label{2.5}
c=\frac{h^2}{\Delta^2(c)}+\frac{1}{V}\sum_{q\in\Lambda^*}\frac{\lambda}{2\O(c)}
\coth{\frac{\beta\lambda}{2}\O(c)} \ .
\end{equation}
The spectrum $\Omega_q(\C)$, $q\in\Lambda^*$, of $\H$ is defined by the harmonic spectrum $\omega_q$ and
by the gap $\Delta(\C)$:
$$\Omega_q^2(\C):=\Delta(\C)+\omega_q^2 \ , $$
$$\Delta(\C):=a+2W'(\C) \ , $$
$$\omega_q^2:=:\tilde{\phi}(0)-\tilde{\phi}(q) \ , \
\tilde{\phi}(q) := \sum_{l\in\Lambda}\phi_{l,0}\exp{(-iql)} \ . $$
Finally, $\lambda = {\hbar}/{\sqrt{m}}$ is the quantum parameter of the model and $\beta=(k_B T)^{-1}$,
where $T$ is the temperature.

The approximating Hamiltonian method gives for $\H \geq 0$ the following condition of \textit{stability}
in thermodynamic limit $\Lambda\rightarrow\mathbb{Z}$:
\begin{equation}\label{2.12}
\Delta(c_{h}(T,\lambda))=\lim_\Lambda \Delta(c_{\Lambda, h}(T,\lambda)) \geq 0 \ , \
c_{h}(T,\lambda):= \lim_\Lambda c_{\Lambda, h}(T,\lambda) \ .
\end{equation}
Let $a>0$ and $W:\mathbb{R}_{+}^1\rightarrow \mathbb{R}_{+}^1$ be a monotonous decreasing function with
$W''(c)\geq w>0$.  Then by definition of the gap $\Delta(\C)$ and by (\ref{2.12}) one gets for the
\textit{stability domain}:  $D=[c^*,\infty)$, where
$c^*=\inf\{{c\;: c\geq0\;,\;\Delta(c)\geq 0\}}$ and $\Delta(c^*)= a+2W'(c^*) = 0 $.
\begin{theo} \label{thm:3.1}
\begin{equation}\label{AHM}
\lim_\Lambda f_{\Lambda}[H_\Lambda^{SCP}] = \lim_\Lambda \sup_{c \, \geq \, c^*} f_{\Lambda}[H_\Lambda (c)]
=: f(\beta, h).
\end{equation}
\end{theo}
By this main for the AHM theorem \cite{VZ1} thermodynamics of the system $H_\Lambda^{SCP}$ and
$H_\Lambda (c)$ for $c =\C$ (\ref{2.5}) are \textit{equivalent}.
Therefore, to study the phase diagram of the model (\ref{2.3}) we have to consider equation
(\ref{2.5}) in the thermodynamic limit $\Lambda\rightarrow\mathbb{Z}$:
\begin{equation}\label{2.14}
c_{h}(T,\lambda)=\rho(T,\lambda, h)+I_d(c_{h}(T,\lambda),T,\lambda) \ .
\end{equation}
Here we split the thermodynamic limit of the integral sum (\ref{2.5}) into zero-mode term plus $h$-term
and the rest:
\begin{eqnarray}\label{2.15}
&&\rho(T,\lambda, h) = \lim_{\Lambda} \rho_{\Lambda}(T,\lambda, h) :=\\
&&\lim_{\Lambda}\left\{\frac{h^2}{\Delta^2(\C)} + \frac{1}{V}\frac{\lambda}{2\sqrt{\Delta(\C)}}
\coth\frac{\beta\lambda}{2}\sqrt{\Delta(\C)}\right\} ,  \nonumber \\
&&I_d(c_{h}(T,\lambda),T,\lambda):= \frac{\lambda}{(2\pi)^d}\int_{q\in B_d}d^dq
\frac{1}{2\Omega_q (c_{h}(T,\lambda))}\coth\frac{\beta\lambda}{2}\Omega_q(c_{h}(T,\lambda)) \ . \nonumber
\end{eqnarray}
Here, $B_d=\{q\in\reel^d;\vert q\vert\leq\pi\}$ is the first Brillouin zone.

To analyse solution of (\ref{2.14}) we consider below two cases:  (a) $h=0$ and (b) $h\not = 0$.

\smallskip

\noindent (a) $h=0$:  From (\ref{2.14}), (\ref{2.15}), one easily gets that for $T=0$, there is
$\lambda_c$ such that $c^\ast\leq I_d(c^\ast,0,\lambda)$ for
$\lambda\geq\lambda_c$ and $c^\ast\ = I_d(c^\ast,0,\lambda_c)$ defines the critical value of the
\textit{quantum parameter} $\lambda$. Then the line $(\lambda, T_c(\lambda))$ of critical temperatures:
$\lambda \mapsto T_c(\lambda)$, which separates the \textit{phase diagram} $(\lambda, T)$ into two domains
(A)-(B), verifies the identity:
\begin{equation}\label{2.17}
c^*=I_d(c^*,T_c(\lambda),\lambda)\, , \ \ \lambda\leq\lambda_c \ \ {\rm{and}} \ \ T_c(\lambda_c ) = 0 \ .
\end{equation}

Taking into account (\ref{2.14}) and (\ref{2.15}) one can express the conditions (\ref{2.17}) as the
\textit{critical-line} equation:
\begin{equation}\label{2.16}
\rho_{c^*}(T_c(\lambda),\lambda):= \rho(T,\lambda, h)\big|_{c_{h}(T,\lambda) = c^*} =
c^*-I_d(c^*,T_c(\lambda),\lambda) = 0 \ .
\end{equation}
Therefore, we obtain two solutions of (\ref{2.14}) distinguished by the value of the gap (\ref{2.12}):
\begin{itemize}
\item  (A) $ \ \rho(T,\lambda, 0) = 0 \ , \ \ c_{0}(T,\lambda) > c^\ast \ \ {\rm{or}} \ \
\Delta(c_{0}(T,\lambda)) > 0 : \,  T > T_c(\lambda) \vee \lambda > \lambda_c$  \ ,
\item  (B) $ \ \rho(T,\lambda, 0) \geq 0 \ , \ \ c_{0}(T,\lambda)=c^\ast \ \ {\rm{or}} \ \
\Delta(c_{0}(T,\lambda)) = 0 : \, 0 \leq T \leq T_c(\lambda) \wedge \lambda \leq \lambda_c$  \ .
\end{itemize}

For $\lambda < \lambda_c$ fixed, by looking along the vertical ($\lambda= \, $ const) line, we observe
the well-known temperature-driven phase transition at $T_c(\lambda)>0$ with  \textit{order} parameter,
which can be identified with $\rho$. On the other hand, for a fixed $T < T_c(0)$, looking along the horizontal
($T = \, $ const) line
one observes a phase transition at $\{\lambda: T_c(\lambda) = T\}$, which is driven by the quantum parameter
$\lambda = \hbar/\sqrt{2m}$.

Note that for $\lambda>\lambda_c$, i.e.  for \textit{light} atoms, the temperature-driven phase
transition is suppressed by quantum tunneling or quantum fluctuations. Decreasing of $T_c(\lambda)$ for
light atoms is well-known as \textit{isotopic effect} in ferroelectrics \cite{APS}. Since by
Theorem \ref{thm:3.1} thermodynamics of the models (\ref{2.3}) and approximating Hamiltonian $\H$ are
are equivalent the proof that one has the same effect in the model (\ref{2.3}) including the existence of
$\lambda_c$ follows from solution of equation (\ref{2.3}) and monotonicity of
$\lambda \mapsto I_d(c^\ast, 0,\lambda)$. The proof of the isotopic effect for the original model (\ref{2.1})
was obtained in \cite{VZ2}, see also \cite{MPZ}, \cite{AKKR}.

To proceed we introduce for Hamiltonians (\ref{2.2}), (\ref{2.3}), and (\ref{2.32}) the canonical Gibbs states:
\begin{equation}\label{2.171}
\omega_{\beta,\Lambda, \ast}(\cdot) = \frac{{\rm{Tr}}_{{\mathcal{H}}_{\Lambda}}
[\exp(-\beta H_{\Lambda,  \ast})\ (\cdot) \ ]}
{{\rm{Tr}}_{{\mathcal{H}}_{\Lambda}} \exp(-\beta H_{\Lambda,   \ast})} \ \ , \ \
H_{\Lambda,   \ast} = H_{\Lambda} \vee  H_{\Lambda}^{SCP} \vee H_{\Lambda}(c) \ .
\end{equation}

Note that by (\ref{2.171}) these states inherit for $h=0$ the $\relatif^2$-symmetry of Hamiltonians
(\ref{2.2}), (\ref{2.3}), and (\ref{2.32}):  $Q_l\rightarrow -Q_l$, i.e. one has
\begin{equation}\label{2.18}
\omega_{\beta,\Lambda, \ast}(Q_l) = \lim_\Lambda \omega_{\beta,\Lambda, \ast}(- Q_l) = 0 \ .
\end{equation}

\noindent (b) $h\not = 0$: Then we obtain
\begin{equation}\label{2.19}
\omega_{\beta, c_h}(Q_l)=\frac{h}{\D(c_h(T,\lambda))} \ .
\end{equation}
For \textit{disordered} phase (A), we have $\lim_{h\rightarrow 0} c_h(T,\lambda)=c(T,\lambda)> c^*$.
So, $\D(c)>0$ and
\begin{equation}\label{2.20}
\lim_{h\rightarrow 0}\omega_{\beta, c_h}(Q_l)=0 \ .
\end{equation}
For \textit{ordered} phase (B), we have $\lim_{h\rightarrow 0} c_h(T,\lambda)=c^*$, then by (\ref{2.15})
\begin{equation}\label{2.21}
\rho_{c^*}(T,\lambda)=c^*-I_d(c^*,T,\lambda)=\lim_{h\rightarrow 0}\frac{h^2}{\D^2(c_h)}>0 \ .
\end{equation}
Finally, (\ref{2.19}) and (\ref{2.21}) yield the values of the physical order parameters
\begin{equation}\label{2.22}
\omega_{\beta, \pm}(Q_l):= \lim_{h\rightarrow\pm 0}\omega_{\beta, c_h}(Q_l)=
\pm\sqrt{\rho_{c^*}(T,\lambda)}\not=0 \ .
\end{equation}

Therefore, using the Bogoliubov quasi-average (\ref{2.22}), and Section \ref{sec:gBEC-BQ-A}, we obtain
two extremal translation invariant equilibrium states $\omega_{\beta, +}$ and $\omega_{\beta, -}$,
invariant by translations, such that
\begin{equation}\label{2.23}
\omega_{\beta, +}(Q_l) = - \ \omega_{\beta, -}(Q_l)=[\rho_{c^*}(T,\lambda)]^{{1}/{2}}\not=0 \ , \ \ \
l\in \mathbb{Z} \ .
\end{equation}
In this case, one can easily check that positions and momenta have \textit{normal}
fluctuations $\delta_Q = \delta_P = 0$ (\ref{1.2}), \cite{VZ1}. We return to this observation below in a
framework of a more general approach: a {scaled} Bogoliubov quasi-averages \cite{VZ1}.
\begin{defi} \label{def:3.11}
We say that external source in (\ref{2.2}), (\ref{2.3}), and (\ref{2.32}) corresponds to the \textit{scaled}
Bogoliubov quasi-average $h\rightarrow 0$, if it is coupled with the thermodynamic limit
$\Lambda\uparrow\mathbb{Z}$ by the relation:
\begin{equation}\label{2.24}
h_{\alpha}:=\frac{\widehat{h}}{V^\alpha} \ , \ \alpha>0 \ .
\end{equation}
\end{defi}
This choice of quasi-average is flexible enough to scan between weak/strong external sources as a function of
$0< \alpha$. It is a message of the following proposition, see \cite{VZ1}.
\begin{proposition}\label{prop:3.12}
If $\alpha<1$, then limiting equilibrium states rest \textit{pure}:  $\lim_\Lambda
\omega_{\beta, c_h}(Q_l)=(sign \, \widehat{h})[\rho_{c^*}(T,\lambda)]^{{1}/{2}}$, which is similar to the case of
the standard Bogoliubov quasi-average $h\rightarrow \pm 0$, (\ref{2.23}).\\
If $\alpha\geq1$, then the limiting state $\omega_{\beta, \widehat{h}}(Q_l)$ becomes a mixture of pure states:
\begin{equation*}
\omega_{\beta, \widehat{h}}(Q_l)= a \ \omega_{\beta, +}(Q_l) + (1-a) \ \omega_{\beta, -}(Q_l) \ ,
\end{equation*}
where $a := a(\widehat{h}, \alpha, \rho_{c^*}(T,\lambda)) \in [0,1]$ is
\begin{equation}\label{2.24a}
 a(\widehat{h}, \alpha, \rho) = \frac{1}{2} \left(1 + \frac{\widehat{h}}{\xi \sqrt{\rho}}\right) \ , \ {\rm{for}} \
 \alpha = 1\ , \ \ \ {\rm{and}} \ \  \  \  a(\widehat{h}, \alpha, \rho) = 1/2 \ , \ {\rm{for}} \  \alpha > 1 \ .
\end{equation}
Here $\xi: = \lim_{\Lambda} [\Delta(\C) \, V] = (2\beta \rho)^{-1} + \sqrt{(2\beta \rho)^{-2} +
\widehat{h}^2/\rho}\ $.
\end{proposition}

Our next step is to study the impact of the scaled quasi-average sources on the quantum fluctuation operators.
Consider now the zero-mode ($k=0$, (\ref{1.2})) fluctuation operators of position and momentum given by
\begin{equation}\label{2.25}
F_{\d_Q}(Q)=\lim_\L\frac{1}{V^{\frac{1}{2}+\d_Q}}\sum_{i\in\L}(Q_i-\omega_{\beta, \Lambda, c_h}(Q_i)) \ ,
\end{equation}
and
\begin{equation}\label{2.26}
F_{\d_P}(P)=\lim_\L\frac{1}{V^{\frac{1}{2}+\d_P}}\sum_{i\in\L}(P_i-\omega_{\beta, \Lambda, c_h}(P_i)) \ .
\end{equation}
Since the approximating Hamiltonian is quadratic operator form (\ref{2.32}), one can calculate the variances
of fluctuation operators (\ref{2.25}) and (\ref{2.26}) explicitly:
\begin{equation*}
\lim_\L\omega_{\beta, \Lambda, c_h}\left(\{\frac{1}{V^{\frac{1}{2}+\d_Q}}
\sum_{i\in\L}(Q_i-\omega_{\beta, \Lambda, c_h}(Q_i))\}^2\right)=
\end{equation*}
\newline
\begin{equation}\label{2.27}
= \lim_\L\frac{1}{V^{2\delta_Q}}\frac{\lambda}{2\sqrt{\Delta(c_{\Lambda, h}(T,\lambda))}}
\coth\frac{\beta\lambda}{2}\sqrt{\Delta(c_{\Lambda, h}(T,\lambda))} \ ,
\end{equation}
\begin{equation*}
\lim_\L\omega_{\beta, \Lambda, c_h}\left(\{\frac{1}{V^{\frac{1}{2}+\d_P}}
\sum_{i\in\L}(P_i-\omega_{\beta, \Lambda, c_h}(P_i))\}^2\right)=
\end{equation*}
\newline
\begin{equation}\label{2.28}
= \lim_\L\frac{1}{V^{2\delta_P}}\frac{\lambda
m\sqrt{\Delta(c_{\Lambda, h}(T,\lambda))}}{2}\coth\frac{\beta\lambda}{2}\sqrt{\Delta(c_{\Lambda, h}(T,\lambda))} \ .
\end{equation}
Here $c_h : = \C$ is a solution of the self-consistent equation (\ref{2.5}) and by the $\relatif^2$-symmetry
of Hamiltonian (\ref{2.3}):  $P_l\rightarrow - P_l$, one has $\omega_{\beta,\Lambda,  c_h}(P_l) = 0$ in
(\ref{2.28}) for all $l \in \Lambda$ and for any values of $\beta , h$.

We note that existence of \textit{nontrivial} variances (\ref{2.27}) and (\ref{2.28}) is sufficient for the
proof of existence of the characteristic function (\ref{1.3}) with sesquilinear form
$S_\omega(\cdot,\cdot)$. The next ingredient is the corresponding to the fluctuation operator algebra
{symplectic} form $\sigma_\omega(\cdot,\cdot)$ one has to calculate the limit of commutator (\ref{1.6}). By
(\ref{2.25}) and (\ref{2.26}) we get
\begin{equation}\label{1.61}
\lim_\Lambda[F^{\d_P}_\Lambda(P),F^{\d_Q}_\Lambda(Q)]=
\lim_\Lambda \frac{1}{V^{1+\d_P + \d_Q}}\sum_{l,l'\in\L} [P_{l} , Q_{l'}] =
\lim_\Lambda \frac{1}{V^{\d_P + \d_Q}} \ \frac{\hbar}{i} \ .
\end{equation}

We summarise this subsection by the following list of comments and remarks.
\begin{rem}\label{rem:3.12} We summarise this subsection by the following list of comments.

(a)-(A) Let $h = 0$ and let $[0, \lambda_c] \ni \lambda \mapsto T_c (\lambda)$.
If the point $(\lambda, T)$ on the phase diagram is above the critical line $(\lambda, T_c (\lambda))$:
$T > T_c(\lambda)$, or if $\lambda > \lambda_c$, see (\ref{2.17}), then this is the case (A), when
$\Delta(c_{h=0}(T,\lambda)) > 0$. Consequently (\ref{2.27}) and (\ref{2.28}) yield $\d_Q=\d_P = 0$,
to ensure  non-triviality of the variances, i.e. of the central limit
both for momentum and for displacement fluctuation operators. They are called \textit{normal}, or
\textit{noncritical} fluctuation operators. Since in this case the commutator (\ref{1.61}) is nontrivial,
the operators $F_{0}(P)$ and $F_{0}(Q)$ are generators of \textit{non-abelian} algebra of normal
fluctuations. Since in this domain of the phase diagram the order parameter $\rho(T,\lambda, h=0) =0$
(\ref{2.15}), we call this pure phase disordered. Note that $\rho(T,\lambda, h=0) =0$ implies
$\omega_{\beta,\Lambda, c_{h=0}}(Q_l) = 0$ even without the reference on $Z_2$-symmetry (\ref{2.15}).

(a)-(B) Let $h = 0$. If the point $(\lambda, T)$ on the phase diagram is below the critical line
$(\lambda, T_c (\lambda))$: $T < T_c(\lambda)$ and $\lambda < \lambda_c$, then
$\lim_\Lambda c_{\Lambda, h=0}(T,\lambda)= c^*$, i.e. the gap
$\lim_\Lambda \Delta(c_{\Lambda, h=0}(T,\lambda)) = 0$ (\ref{2.12}) and the order parameter
$\rho(T,\lambda, h=0) > 0$ (\ref{2.15}). Therefore, by (\ref{2.27}) one gets $\delta_Q = 1/2$ and by
(\ref{2.28}) one gets $\delta_P = 0$ to ensure a nontrivial central limit. Hence, the
displacement fluctuation operator $F_{1/2}(Q)$ is \textit{abnormal}, whereas the momentum fluctuation operator
$F_{0}(P)$ is normal. By (\ref{1.61}) the operators $F_{1/2}(Q), F_{0}(P)$  (\ref{2.25}), (\ref{2.26}),
commute, i.e.  they generate a \textit{abelian} algebra of fluctuations. We comment that although
order parameter $\rho(T,\lambda, h=0) > 0$ the $Z_2$-symmetry (\ref{2.15}) implies that \textit{displacement}
order parameter $\omega_{\beta, c^*}(Q_l) = 0$. The Bogoliubov quasi-average (\ref{2.22})
gives non-zero value for \textit{displacement} order parameter. This means that $\omega_{\beta, c^*}$
is the \textit{one-half} mixture of the pure states $\omega_{\beta,\pm}$  (\ref{2.23}) and explains
abnormal fluctuation of displacement.

Now let $h \neq 0$ and consider the standard Bogoliubov quasi-averages, see (b).

(b)-(A) Since $\Delta(c_{h}(T,\lambda)) > 0$, by (\ref{2.27}), (\ref{2.28}) one gets the finite
quasi-averages
\begin{eqnarray*}
&&\lim_{h \rightarrow 0} \lim_\L\omega_{\beta, \Lambda, c_h}\left(\{\frac{1}{V^{\frac{1}{2}}}
\sum_{i\in\L}(Q_i-\omega_{\beta, \Lambda, c_h}(Q_i))\}^2\right) \ , \\
&&\lim_{h \rightarrow 0} \lim_\L\omega_{\beta, \Lambda, c_h}\left(\{\frac{1}{V^{\frac{1}{2}}}
\sum_{i\in\L}(P_i-\omega_{\beta, \Lambda, c_h}(P_i))\}^2\right)\ . \nonumber
\end{eqnarray*}
They yield the same result on the normal fluctuations as in (a)-(A).

(b)-(B) Since $h \neq 0$, the difference with the case (a)-(B) comes from
$\lim_\Lambda \Delta(c_{\Lambda, h}(T,\lambda)) > 0$ (\ref{2.12}) and from (\ref{2.21}), which is valid in
the ordered phase. Then the quasi-average for the displacement variance (\ref{2.27}):
\begin{equation}\label{1.62}
\lim_{h \rightarrow  0}
\lim_\L\frac{1}{V^{2\delta_Q}}\frac{\lambda}{2\sqrt{\Delta(c_{\Lambda, h}(T,\lambda))}}
\coth\frac{\beta\lambda}{2}\sqrt{\Delta(c_{\Lambda, h}(T,\lambda))} \ ,
\end{equation}
has \textit{no} nontrivial sense for any $\delta_Q$. Whereas the quasi-average for the momentum variance
(\ref{2.28}) is nontrivial only when $\delta_P = 0$.

(b*)-(B) This difficulty is one of the motivation to consider instead of (\ref{1.62}) the \textit{scaled}
Bogoliubov quasi-average (\ref{2.24}) for $h_{\alpha}$.

(a)-(B*) We conclude this remark by the case when the point $(\lambda, T)$ belongs to the
\textit{critical line}: $(\lambda, T_c(\lambda))$, where $\lambda \leq \lambda_c$. Therefore, the gap
$\lim_\Lambda \Delta(c_{\Lambda, h=0}(T_c(\lambda),\lambda)) = 0$ (\ref{2.12}) and the order parameter
$\rho(T_c(\lambda),\lambda, h=0) = 0$ (\ref{2.15}).

(i) If $\lambda < \lambda_c$, then $T_c(\lambda) > 0$. Hence, by (\ref{2.28}) the momentum fluctuation operator
is normal, $\d_P = 0$, whereas displacement fluctuation operator is \textit{abnormal} with the power
$\d_Q >0 $, which depends on the asymptotics $\mathcal{O}(V^{-\gamma}), \gamma > 0$,  of the gap
$\Delta(c_{\Lambda, h=0}(T_c(\lambda)))$ in thermodynamic limit.

Note that in the \textit{scaled} limit $\lim_\Lambda \Delta(c_{\Lambda, h_{\alpha}}(T_c(\lambda),\lambda)) = 0$
the asymptotics $\mathcal{O}(V^{-\gamma})$, and by consequence $\d_Q >0 $, may be modified by the power
$\alpha$. Although it leaves stable $\d_P = 0$. We study this phenomenon in the next section. By (\ref{1.61})
the corresponding algebra of fluctuations is abelian.

(ii) If $\lambda = \lambda_c$, then $T_c(\lambda_c) = 0$ (\ref{2.15}) and one observes a zero-temperature
\textit{quantum} phase transition at the critical point $(0,\lambda_c)$ by varying the quantum parameter
$\lambda$. In this case the variances (\ref{2.27}), (\ref{2.28}) take the form
\begin{equation}\label{2.271}
\lim_\L\frac{1}{V^{2\delta_Q}}\frac{\lambda_c}{2\sqrt{\Delta(c_{\Lambda, h}(0,\lambda_c))}} \ ,
\end{equation}
\begin{equation}\label{2.281}
\lim_\L\frac{1}{V^{2\delta_P}}\frac{\lambda_c m\sqrt{\Delta(c_{\Lambda, h}(0,\lambda_c))}}{2}\ .
\end{equation}
Since $\Delta(c_{\Lambda, h_{\alpha}}(0,\lambda_c)) = \mathcal{O}(V^{-\gamma})$, (\ref{2.271})
implies that the displacement fluctuation operator is abnormal with the power $\d_Q = \gamma/4 >0 $,
which may be modified by the power $\alpha$. The momentum fluctuation operator is also abnormal, but
\textit{squeezed} since by (\ref{2.281}) one gets  $\d_P = - \gamma/4 < 0$ . Note that $\d_Q + \d_P  = 0$
yields a nontrivial commutator (\ref{1.61}). Therefore, algebra of abnormal fluctuations generated by
$F_{\d_Q}(Q), F_{\d_P}(P)$ is non-abelian and possibly $\alpha$-dependent.

\end{rem}

In the next Sections we elucidate a relation between definition of quantum fluctuation
operators and the scaled Bogoliubov quasi-averages (\ref{2.24}) indicated in Remark \ref{rem:3.12}.

\section{Quasi-averages for critical quantum fluctuations} \label{Q-A-Cr-Q-Fl}
\subsection{Quantum fluctuations below the critical line} \label{QF-below}
We consider here the case (b*)-(B). We show that the \textit{scaled} Bogoliubov quasi-average (\ref{2.24})
for $h_{\alpha}$ is relevant for analysis of fluctuations below the critical line.
\begin{proposition}\label{prop:3.13}
Let $0 \leq T < T_c(\lambda) \wedge \lambda < \lambda_c$. Then the momentum fluctuation operator
is normal, $\d_P = 0$, whereas displacement fluctuation operator is \textit{abnormal} with the power
$0 < \d_Q \leq 1/2 $, which depends on the scaled Bogoliubov quasi-average parameter ${\alpha}$ (\ref{2.24}).
The fluctuation algebra is abelian.
\end{proposition}
\begin{proof}
(1) Let $0<\alpha<1$. Then by (\ref{2.15}), (\ref{2.21}), and (\ref{2.24}) we obtain that
\begin{equation}\label{2.272-1}
\omega_{\beta, {\rm{sign}}(\widehat{h})}(Q_l) = \lim_{\Lambda}\omega_{\beta,\Lambda, c_h}(Q_l)= \lim_{\Lambda}
\frac{\widehat{h}}{V^{\alpha}\Delta(c_{\Lambda,h}(T,\lambda))} =
{\rm{sign}}(\widehat{h})\sqrt{\rho_{c^*}(T,\lambda)} \ .
\end{equation}
This indicates that this scaled quasi-average limit gives pure states (\ref{2.23}) and that by (\ref{2.27})
the variance of displacement fluctuation operator has a finite value
\begin{equation}\label{2.272}
0 < \lim_\L\omega_{\beta, \Lambda, c_h}\left(\{\frac{1}{V^{\frac{1}{2}+\d_Q}}
\sum_{i\in\L}(Q_i-\omega_{\beta, \Lambda, c_h}(Q_i))\}^2\right)=
\lim_\L\frac{1}{V^{2\delta_Q - \alpha}}\frac{\sqrt{\rho_{c^*}(T,\lambda)}}
{2 \beta |\widehat{h}|}  < \infty ,
\end{equation}
if $\delta_Q = \alpha/2$. On the other hand, the finiteness of (\ref{2.28}) implies that $\delta_P = 0$,
i.e. the momentum fluctuation operator is normal. Since $0<\alpha$, by (\ref{1.61}) the fluctuation algebra is
abelian.

(2) Let $\alpha = 1$. Then by (\ref{2.14}), (\ref{2.15}), and (\ref{2.24}) we obtain that
\begin{eqnarray}\label{2.273}
\rho_{c^*}(T,\lambda) &=& \lim_{\Lambda} \rho_{\Lambda}(T,\lambda, h) =
\lim_{\Lambda} \left\{\frac{{\widehat{h}}^2}{[V \Delta(c_{\Lambda,h}(T,\lambda))]^2} +
\frac{1}{V \Delta(c_{\Lambda,h}(T,\lambda))}\right\} \\
&=& c^* - I_d(c^* ,T,\lambda) > 0 \ , \nonumber
\end{eqnarray}
yields the bounded $w_{\widehat{h}}((T,\lambda))) : = \lim_{\Lambda} [V \Delta(c_{\Lambda,h}(T,\lambda))] > 0$ for
$h = \widehat{h}/V$. Then the displacement order parameter
\begin{equation*}
- \sqrt{\rho_{c^*}(T,\lambda)} < \lim_{\Lambda}\omega_{\beta,\Lambda, c_h}(Q_l)= \lim_{\Lambda}
\frac{\widehat{h}}{V \Delta(c_{\Lambda,h}(T,\lambda))} = \frac{\widehat{h}}{w_{\widehat{h}}((T,\lambda)))}
< \sqrt{\rho_{c^*}(T,\lambda)} \ .
\end{equation*}
This means that the equilibrium Gibbs state
\begin{equation}\label{2.274}
\omega_{\beta, \widehat{h}}(\cdot) = \xi \ \omega_{\beta, +}(\cdot) +
(1 - \xi) \ \omega_{\beta, -}(\cdot) \ , \ \ \xi = \frac{1}{2}[1 +
{\widehat{h}}/{(w_{\widehat{h}}(T,\lambda)\sqrt{\rho_{c^*}(T,\lambda)})}] \in (0,1) \ .
\end{equation}
is a convex combination the pure states (\ref{2.272-1}). Note that (\ref{2.27}) and the boundedness of
$w_{\widehat{h}}((T,\lambda)))$ imply: $\delta_Q = 1/2$, whereas (\ref{2.28}) gives $\delta_P = 0$ .
So, in the mixed state $\omega_{\beta, \widehat{h}}(\cdot)$ the displacement fluctuations are abnormal,
but the momentum fluctuation operator rests normal, and the fluctuation algebra is abelian as in the
case (1).

(3) Let $\alpha > 1$. Then again by (\ref{2.14}), (\ref{2.15}), and by (\ref{2.24}) we obtain that
\begin{equation}\label{2.275}
\rho_{c^*}(T,\lambda) =
\lim_{\Lambda} \frac{1}{V \Delta(c_{\Lambda,h}(T,\lambda))} = c^* - I_d(c^* ,T,\lambda) > 0 \ ,
\end{equation}
which by (\ref{2.19}) yields for the displacement order parameter
\begin{equation}\label{2.276}
\lim_{\Lambda}\omega_{\beta,\Lambda, c_h}(Q_l)= \lim_{\Lambda}
\frac{\widehat{h}}{V^{\alpha} \Delta(c_{\Lambda,h}(T,\lambda))} = 0 \ .
\end{equation}
Note that these scaled quasi-averages (\ref{2.275}), (\ref{2.276}) in the ordered phase (B) are completely
different from the standard quasi-average (\ref{2.21}), (\ref{2.22}). By (\ref{2.27}) and by (\ref{2.28})
one gets $\delta_Q = 1/2$ and $\delta_P = 0$, which are the same as in the case (2), including the abelian
fluctuation algebra.
We comment that the case $\alpha > 1$ is formally equivalent to the case (2) for $\widehat{h} \rightarrow 0$,
which implies $\xi \rightarrow 1/2$, see (\ref{2.274}). The same one deduce from (\ref{2.276}).
\end{proof}
\subsection{Abelian algebra of fluctuations on the critical line} \label{Ab-alg-fl}

In this section we are going to characterise the exponents $\d_Q$ and $\d_P$ on the critical line as
function of parameters $d,\sigma$ and (if it is the case) of the parameter $\alpha$. To this
end, we proceed as follows. Note that the critical line is defined by equation (\ref{2.16}).
Hence, $\rho_{c^\ast}(T_c(\lambda),\lambda)=0$, and (\ref{2.15}) for $\lim_\Lambda \Dc = c^\ast$ takes
the form
\begin{equation}\label{3.4}
\lim_{\Lambda}\left\{\frac{1}{V}\frac{\lambda}{2\sqrt{\Dc}}\coth\frac{\beta_{c}\lambda}{2}\sqrt{\Dc}+
\frac{\widehat{h}^2}{V^{2\alpha}\Dc^2}\right\}=0 \ ,
\end{equation}
where $\beta_{c} := (k_B T_c(\lambda))^{-1}$.

Since for the scaled quasi-average (\ref{2.24}) we choose $h=\widehat{h}/V^\alpha$, by (\ref{2.15}) and
(\ref{2.16}) the limit $\lim_\Lambda c_{\Lambda,h}(T_c(\lambda),\lambda) = c^\ast$. Hence, $\lim_\Lambda\Dc =0$.
Now one has to distinguish two cases: \\ \\
\hspace*{2 cm}(a)
$T_c(\lambda)>0$ (\ref{2.17}), then (\ref{3.4}) is equivalent to
\begin{equation}\label{3.6}
\lim_{\Lambda}\left\{\frac{1}{V\Dc\beta_c}+
\frac{\widehat{h}^2}{V^{2\alpha}\Dc^2}\right\}=0 \ ,
\end{equation}
\hspace*{2cm}(b) $T_c(\lambda_c)=0$ (\ref{2.17}), then (\ref{3.4}) is equivalent to
\begin{equation}\label{3.7}
\lim_{\Lambda}\left\{\frac{\lambda}{2V\sqrt{\Delta(c_{\Lambda,h}(0,\lambda))}}+
\frac{\widehat{h}^2}{V^{2\alpha}\Delta(c_{\Lambda,h}(0,\lambda))^2}\right\}=0 \ .
\end{equation}
Both cases imply that for $V\rightarrow\infty$ the gap $\Delta$ in (\ref{3.4}) has the asymptotic
behaviour $\Delta\simeq V^{-\gamma}$ with $(0<\gamma<1) \wedge (0<\gamma<\alpha)$ for (\ref{3.6}) or
$(0<\g<2) \wedge (0<\g<\alpha)$ for (\ref{3.7}), correspondingly.

Note that it is equation (\ref{2.5}), which is the key to calulate these asymptotics. To make this argument
evident, we rewrite (\ref{2.5}) identically as
\begin{eqnarray}
&&(c_\Lambda-c^\ast)+[c^\ast-I_d(c_\Lambda,T_c(\lambda),\lambda)]+\left[I_d(c_\Lambda,T_c(\lambda),\lambda)-
\frac{1}{V} \sum_{q\in\Lambda^\ast,q\not=0}\frac{\lambda}
{2\Omega_q(c_\Lambda)}\coth\frac{\beta_c \lambda\Omega_q(c_\Lambda)}{2}\right] \nonumber \\
&&= \left(\frac{\widehat{h}}{V^\alpha\Delta}\right)^2+\frac{1}{V}\frac{\lambda}{2\sqrt{\Delta}}
\coth\left(\frac{\beta_c\lambda \sqrt{\Delta}}{2}\right) \ ,   \label{3.8}
\end{eqnarray}
here we denote $c_\Lambda := c_{\Lambda,h}(T_c(\lambda),\lambda)$ and $\Delta := \Dc$ .
The asymptotic behaviour of the left-hand side of equation (\ref{3.8}) resulats from the hypothesis
$h=\widehat{h}/V^\alpha$ and from the convergence rate of the Darboux-Riemann sum to the limit of integral
$I_d(c_\Lambda,T_c(\lambda),\lambda)$. Together with asymptotics of the right hand-side this gives the power
$\gamma$.
\begin{proposition}\label{proposition:5.1} If ($T,\lambda$) belongs to the
critical line ($T_c(\lambda),\lambda$) with $T_c(\lambda)>0$, then the
asymptotic volume behaviour of the gap $\Dc$ is defined by \\ \begin{math}
\hspace*{2,5 cm}\gamma=\left\{\begin{array}{ll} if\; d>2\sigma\\ \hspace{1
cm}\gamma=\frac{2}{3}\alpha \hspace{1 cm}&for\; \alpha<\frac{3}{4}=\alpha_c\\
\hspace{1 cm}\gamma=\frac{1}{2}\hspace{1 cm}&for\; \alpha\geq\frac{3}{4}\\ \\
if\; d=2\sigma\\ \hspace{1 cm}\gamma=\frac{2}{3}\alpha+0 \hspace{1 cm}&for\;
\alpha<\frac{3}{4}=\alpha_c\\ \hspace{1 cm}\gamma=\frac{1}{2}+0 \hspace{1
cm}&for\; \alpha\geq\frac{3}{4}\\ \\ if\;\sigma<d<2\sigma\\ \hspace{1
cm}\gamma=2\alpha\frac{\sigma}{d+\sigma}\hspace{1 cm}&for\;
\alpha<\frac{1}{2}+\frac{\sigma}{2d}=\alpha_c\\ \hspace{1
cm}\gamma=\frac{\sigma}{d}\hspace{1 cm}&for\;
\alpha\geq\frac{1}{2}+\frac{\sigma}{2d} \end{array} \right.\ \end{math}
\end{proposition}
Since $T_c(\lambda)>0$, the right side of (\ref{3.8}) has asymptotics (\ref{3.6}) or
\begin{equation}\label{3.9}
O [(V \Delta)^{-1} + (V^\alpha \Delta)^{-2} ] \ .
\end{equation}
Let us define $\alpha_c$ such that $O[(V \Delta)^{-1}] = O[(V^{\alpha_c}\Delta)^{-2}]$
Then for the asymptotics (\ref{3.9}) one obviously gets
$O [(V \Delta)^{-1} + (V^{\alpha_c} \Delta)^{-2}] = O [(V \Delta)^{-1}]$,
i.e. for $\alpha = \alpha_c$, the gap $\Delta$ has the same asymptotic behaviour as for
$\widehat{h}=0$. The three regimes of the potentiel decreasing $\sigma$ indicated in Proposition
\ref{proposition:5.1} are considered in details in \cite{JZ98}.
\begin{theo} \label{thm:5.1} If ($T,\lambda$) belongs to the critical line
($T_c(\lambda),\lambda$) with $T_c(\lambda)>0$, then the algebra of
fluctuation operators is abelian. The momentum fluctuation operator
$F_{\d_P}(P)$ is normal ($\d_P=0$) while the position fluctuation operator
$F_{\d_Q}(Q)$ is abnormal with a critical exponent given by\\
\\
\begin{math}
\hspace*{2,5 cm}\d_Q=\left\{\begin{array}{ll}
               if\; d>2\sigma\\
              \hspace{1 cm}\d_Q=\frac{1}{3}\alpha \hspace{1 cm}&for\; \alpha<\frac{3}{4}=\alpha_c\\
              \hspace{1 cm}\d_Q=\frac{1}{4}\hspace{1 cm}&for\; \alpha\geq\frac{3}{4}\\
\\
               if\; d=2\sigma\\
               \hspace{1 cm}\d_Q=\frac{1}{3}\alpha+0 \hspace{1 cm}&for\; \alpha<\frac{3}{4}=\alpha_c\\
               \hspace{1 cm}\d_Q=\frac{1}{4}+0 \hspace{1 cm}&for\; \alpha\geq\frac{3}{4}\\
\\
               if\;\sigma<d<2\sigma\\
               \hspace{1 cm}\d_Q=\alpha\frac{\sigma}{d+\sigma}\hspace{1 cm}&si\; \alpha<\frac{1}{2}+
               \frac{\sigma}{2d}=\alpha_c\\
               \hspace{1 cm}\d_Q=\frac{\sigma}{2d}\hspace{1 cm}&for\; \alpha\geq\frac{1}{2}+\frac{\sigma}{2d}
              \end{array}
        \right.\
\end{math}\
\end{theo}
\begin{proof}
To check the abelian character of the algebra of
fluctuation operators generated by $F^{\d_Q}$ and $F^{\d_P}$, it is enough to note that the limit of the
commutator:
\begin{equation*}
\lim_\Lambda\left[ F_\Lambda^{\d_P},F_\Lambda^{\d_Q}\right]=
\lim_\Lambda\frac{1}{\vert\Lambda\vert^{1+\d_P+\d_Q}}
\sum_{l,l'\in\Lambda}[P_l,Q_{l'}]=0 \ .
\end{equation*}
The second part of the theorem  results from (\ref{2.27}) and (\ref{2.28}), which get on the critical line for
$h=\widehat{h}/V^\alpha$ the form:
\begin{equation}\label{3.22}
\lim_\L\omega_{\beta, \Lambda, c_h}\left(\{\frac{1}{V^{\frac{1}{2}+\d_Q}}
\sum_{i\in\L}(Q_i-\omega_{\beta, \Lambda, c_h}(Q_i))\}^2\right)=
\lim_\L\frac{1}{V^{2\delta_Q}}\frac{kT_c(\lambda)}{\Dc} \ ,
\end{equation}
and
\begin{equation}\label{3.23}
\lim_\L\omega_{\beta, \Lambda, c_h}\left(\{\frac{1}{V^{\frac{1}{2}+\d_P}}
\sum_{i\in\L}(P_i-\omega_{\beta, \Lambda, c_h}(P_i))\}^2\right)=
\lim_\L\frac{1}{V^{2\delta_P}}mkT_c(\lambda) \ .
\end{equation}
So the variance (\ref{3.22}) is not trivial if and only if $\d_Q=\g/2$ and (\ref{3.23}) is not trivial if
and only if $\d_P=0$. Here the value of $\d_Q$ is defined by  Proposition \ref{proposition:5.1}.
\end{proof}

We comment that if one puts in the preceding theorem $\sigma =2$, then the statement corresponds to
short-range interactions when $\sigma\geq2$, see (\ref{eq:1.12}). This result coincides with that in
\cite{VZ1} if one puts $\alpha = \infty$, i.e. when there are no quasi-average sources.
\subsection{Non-abelian algebra of fluctuations on the critical line} \label{Non-abel-fl}
\begin{proposition}\label{proposition:5.2} If ($T,\lambda$) coincides with the critical point $(0,\lambda_c)$,
then the asymptotic volume behaviour of the gap $\Delta(c_{\L,h}(0,\lambda),0)$ is
given by\\ \\ \begin{math} \hspace*{2,5 cm}\gamma=\left\{\begin{array}{ll}
if\; d>\frac{3\sigma}{2}\\ \hspace{1 cm}\gamma=\frac{2}{3}\alpha \hspace{1
cm}&for\; \alpha<1=\alpha_c\\ \hspace{1 cm}\gamma=\frac{2}{3}\hspace{1
cm}&for\; \alpha\geq 1\\ \\ if\; d=\frac{3\sigma}{2}\\ \hspace{1
cm}\gamma=\frac{2}{3}\alpha+0 \hspace{1 cm}&for\;\alpha<1=\alpha_c\\
\hspace{1 cm}\gamma=\frac{1}{2}+0 \hspace{1 cm}&for\; \alpha\geq 1\\ \\
if\;\frac{\sigma}{2}<d<\frac{3\sigma}{2}\\ \hspace{1
cm}\gamma=2\alpha\frac{2\sigma}{2d+3\sigma}\hspace{1 cm}&for\;
\alpha<\frac{1}{2}+\frac{3\sigma}{4d}=\alpha_c\\ \hspace{1
cm}\gamma=\frac{\sigma}{d}\hspace{1 cm}&for\;
\alpha\geq\frac{1}{2}+\frac{3\sigma}{4d} \end{array} \right.\ \end{math}
\end{proposition}
At the point $(0,\lambda_c)$ of the critical line, we obtain the limit (\ref{3.7}), i.e. the gap has the
asymptotic $\Delta \simeq V^{-\gamma}$. Then the right-hand side of (\ref{3.8}) gets the following asymptotic
form:
$O[({V^\alpha\Delta})^{-2}+ (V\Delta^{\frac{1}{2}})^{-1}]$.
Similar to Proposition \ref{proposition:5.1} we define $\alpha = \alpha_c$ in such a way that
$O[({V^\alpha_c \Delta})^{-2}] = O[(V\Delta^{\frac{1}{2}})^{-1}]$. Again one has to consider three
regimes for the value of $\sigma$ as it is indicated in Proposition \ref{proposition:5.2} \cite{JZ98}.
\begin{theo} \label{thm:5.2} If $(T,\lambda)$ coincides with the critical
point $(0,\lambda_c)$, then the algebra of fluctuation operators is
non-abelian because the position fluctuation operator $F_{\d_Q}(Q)$ is
abnormal ($\d_Q>0$), while the momentum fluctuation operator $F_{\d_P}(P)$ is
supernormal (squeezed) with $\d_P=-\d_Q$ and\\
\\
\begin{math} \hspace*{2,5cm}\delta_Q=\left\{\begin{array}{ll} if\; d>\frac{3\sigma}{2}\\
\hspace{1cm}\delta_Q =\frac{1}{6}\alpha \hspace{1 cm}&for\; \alpha<1=\alpha_c\\
\hspace{1 cm}\delta_Q =\frac{1}{6}\hspace{1 cm}&for\; \alpha\geq 1\\ \\ if\;
d=\frac{3\sigma}{2}\\ \hspace{1 cm}\delta_Q =\frac{1}{6}\alpha+0 \hspace{1
cm}&for\; \alpha<1=\alpha_c\\ \hspace{1 cm}\delta_Q =\frac{1}{8}+0 \hspace{1
cm}&for\; \alpha\geq 1\\ \\ if\;\frac{\sigma}{2}<d<\frac{3\sigma}{2}\\
\hspace{1 cm}\delta_Q =\alpha\frac{\sigma}{2d+3\sigma}\hspace{1 cm}&for\;
\alpha<\frac{1}{2}+\frac{3\sigma}{4d}=\alpha_c\\ \hspace{1
cm}\delta_Q =\frac{\sigma}{4d}\hspace{1 cm}&for\;
\alpha\geq\frac{1}{2}+\frac{3\sigma}{4d} \end{array} \right.\ \end{math}
\end{theo}
\begin{proof}
By (\ref{2.17}) the limit $\lim_{\lambda \rightarrow \lambda_c -0} (T_c (\lambda),\lambda) = (0,\lambda_c)$
yields: $\beta_c =(k_B T_c (\lambda))^{-1}\rightarrow\infty$. Then the  variances (\ref{2.27}) and (\ref{2.28})
become
\begin{equation}\label{3.37}
\lim_\L\omega_{\infty, \Lambda, c_h}\left(\{\frac{1}{V^{\frac{1}{2}+\d_Q}}
\sum_{i\in\L}(Q_i-\omega_{\beta, \Lambda, c_h}(Q_i))\}^2\right)=
\lim_\L\frac{1}{V^{2\delta_Q}}\frac{\lambda}{\sqrt{\Delta(c_{\Lambda,h}(0,\lambda_c))}}\ ,
\end{equation}
and
\begin{equation}\label{3.38}
\lim_\L\omega_{\infty, \Lambda, c_h}\left(\{\frac{1}{V^{\frac{1}{2}+\d_P}}
\sum_{i\in\L}(P_i-\omega_{\beta, \Lambda, c_h}(P_i))\}^2\right)=
\lim_\L\frac{1}{V^{2\delta_P}}\frac{\lambda m}{2}\sqrt{\Delta(c_{\Lambda,h}(0,\lambda_c))} \ .
\end{equation}
Since $\Delta \simeq V^{-\gamma}$, one has just to apply Proposition \ref{proposition:5.2} for
$\d_Q=\g/4=-\d_P$ to get the possible values of $\d_Q$. As far as $\d_Q+\d_P=0$. The non-abelian nature of
the algebra of fluctuation operators follows from commutator (\ref{1.61}).
\end{proof}

Note that the same remark about the $\sigma = 2$, as at the end of Section \ref{Ab-alg-fl}, is also valid
for the  quantum critical fluctuations at the point $(0,\lambda_c)$.
\section{Concluding remarks}

In this paper we scrutinise the Bogoliubov method of quasi-averages for quantum systems that manifest
phase transitions.

First, we re-examine a possible application of this method to analysis of the phase transitions with
Spontaneous Symmetry Breaking (SSB). To this aim we consider examples of the Bose-Einstein condensation in
continuous perfect and interacting systems. The existence of different type of generalised condensations leads
to conclusion (see Sections \ref{sec:gBEC-BQ-A} and \ref{sec:BogAppr-Q-A}) that the only
physically reliable quantities are those that defined by the Bogoliubov quasi-averages.

In the second part of the paper we advocate the Bogoluibov method of the \textit{scaled} quasi-averages.
By taking the structural quantum phase transition as a basic example, we scrutinise a relation between SSB
and the critical quantum fluctuations. Our analysis in Section \ref{BQ-A-QFl} shows that again the
scaled quasi-averages give an adequate tool for description of the algebra of quantum fluctuation operators.
The subtlety of quantum fluctuations is already visible on the level of existence-non-existence of the order
parameter that can be \textit{destroyed} by quantum fluctuations even for the zero temperature,
Section \ref{QPT-Fluct-Q-A}. The standard Bogoluibov method is sufficient to this analysis.

A relevance of the scaled Bogoluibov quasi-averages becomes evident for (\textit{mesoscopic}) quantum 
fluctuation operators defined by the Quantum Central Limit since this limit becomes \textit{sensible} to 
the value of the scaling parameter rate $\alpha$. In contract to the non-abelian algebra of \textit{normal} 
quantum fluctuation operators in the disordered phase the \textit{critical} quantum fluctuations in the 
ordered phase and on the critical line do depend on the parameter $\alpha$, see Section \ref{Q-A-Cr-Q-Fl}. 
This concerns \textit{abnormal} and \textit{supernormal} (\textit{squeezed}) quantum fluctuations. They 
manifest variety of abelian-non-abelian algebras of fluctuation operators, Sections 
\ref{QF-below}-\ref{Non-abel-fl}, which are all $\alpha$-dependent.

\section{Appendix A}

In this Appendix we reproduce, for the reader's convenience, the statement of the basic theorem of Fannes,
Pul\`{e} and Verbeure \cite{FPV1}, see also \cite{PVZ} for the extension to nonzero momentum, and Verbeure's
book \cite{Ver}. Unfortunately, neither \cite{FPV1} nor \cite{PVZ} show that the states
$\omega_{\beta,\mu,\phi},\phi \in [0,2\pi)$ in the theorem below are ergodic. The simple, but instructive
proof of this fact was given by Verbeure in his book \cite{Ver}.
\begin{proposition}\label{prop:A.1}
Let $\omega_{\beta,\mu}$ be an analytic, gauge-invariant equilibrium state. If
$\omega_{\beta,\mu}$ exhibits ODLRO (\ref{4.16}), then there exist ergodic states
$ \omega_{\beta,\mu,\phi},\phi \in [0,2\pi)$, not gauge invariant, satisfying {\rm{:}}
(i) $\forall \theta,\phi \in [0,2\pi)$ such that $\theta \ne \phi$, $\omega_{\beta,\mu,\phi} \ne
\omega_{\beta,\mu,\theta}$;
(ii) the state $\omega_{\beta,\mu}$ has the decomposition
\begin{equation*}
\omega_{\beta,\mu} = \frac{1}{2\pi} \int_{0}^{2\pi} d\phi\omega_{\beta,\mu,\phi} \ .
\end{equation*}
(iii) For each polynomial $Q$ in the operators $\eta(b_{{0}})$,$\eta(b_{{0}}^{*})$, and for each
$\phi \in [0,2\pi)$,
\begin{eqnarray*}
 \omega_{\beta,\mu,\phi}(Q(\eta(b_{{0}}^{*}),\eta(b_{{0}})X)
=  \omega_{\beta,\mu,\phi}(Q(\sqrt{\rho_{0}} \exp(-i\phi),\sqrt{\rho_{0}} \exp(i \phi)X)\ \
\forall X \in {\cal A} \ .
\end{eqnarray*}
\end{proposition}

We remark, with Verbeure \cite{Ver}, that the proof of Proposition \ref{prop:A.1} is {constructive}.
One essential ingredient is the separating character (or faithfulness) of the state $\omega_{\beta,\mu}$, i.e.,
$\omega_{\beta,\mu}(A) = 0$ implies $A=0$. This property, which depends on the extension of $\omega_{\beta,\mu}$
to the von-Neumann algebra $\pi_{\omega}({\cal A})^{''}$ (see \cite{BR97}, \cite{Hug}) is true for thermal
states, but is not true for ground states, even without this extension: in fact, a ground state (or vacuum)
is non-faithful on ${\cal A}$ (see Proposition 3 in \cite{Wrep}). We see, therefore, that thermal states and
ground states might differ with regard to the ergodic decomposition (ii). Compare also with our discussion
in the Concluding remarks.

\bigskip

\noindent
\textbf{Acknowledgements}

\noindent
Certain issues dealt with in this manuscript are developed in our paper \cite{WZ16}. Other topics are
originated from the open problems posed in Sec.3 of \cite{SeW} and in \cite{JaZ10} as well as from discussions
around \cite{WZ16}. One of us (W.F.W.) would like to thank G. L. Sewell for sharing with him his views on ODLRO
along several years. He would also like to thank the organisers of the Satellite conference "Operator Algebras
and Quantum Physics" of the XVIII conference of the IAMP (Santiago de Chile) in S\~{a}oPaulo,
July 17th-23rd 2015, for the opportunity to present a talk on these topics.

We are thankful to Bruno Nachtergaele for useful remarks and suggestions concerning the
problems that we discussed in \cite{WZ16} and in the present paper.


\end{document}